\documentclass[12pt]{article}
\usepackage{amsfonts} 
\usepackage{url}  
\usepackage{graphicx}
\usepackage{hyperref} 
\usepackage{float}
\usepackage{color}
\usepackage{amsmath}
\usepackage{amssymb}
\usepackage{amsthm} 
\usepackage{multirow}
\usepackage{subfig} 
\usepackage{algorithm} 
\usepackage{booktabs}
\usepackage{rotating}
\usepackage{algpseudocode}
\usepackage[round]{natbib}
\usepackage{authblk}

\newcommand{\col}{\text{col}}

\newcommand{\tr}{\text{tr}}

\newcommand{\rank}{\text{rank}}
\newcommand{\diag}{\text{diag}}
\newcommand{\vect}{\text{vec}}

\newcommand{\Ex}{\mathbb{E}} 
\newcommand{\real}{\mathbb{R}} 
\newcommand{\In}{\mathbb{I}} 

\def\ba{\boldsymbol{a}}

\def\bu{\boldsymbol{u}}
\def\bv{\boldsymbol{v}}

\def\bx{\boldsymbol{x}}
\def\by{\boldsymbol{y}}

\def\bA{\boldsymbol{A}}
\def\bB{\boldsymbol{B}}

\def\bD{\boldsymbol{D}}
\def\bE{\boldsymbol{E}}

\def\bI{\boldsymbol{I}}
\def\bJ{\boldsymbol{J}}

\def\bL{\boldsymbol{L}}
\def\bM{\boldsymbol{M}}

\def\bP{\boldsymbol{P}}

\def\bR{\boldsymbol{R}}

\def\bU{\boldsymbol{U}}
\def\bV{\boldsymbol{V}}
\def\bW{\boldsymbol{W}}
\def\bX{\boldsymbol{X}}

\newcommand{\0}{\mathbf{0}}
\newcommand{\1}{\mathbf{1}}

\newcommand{\bbeta}{\boldsymbol{\beta}}

\newcommand{\bmu}{\boldsymbol{\mu}}

\newcommand{\bpi}{\boldsymbol{\pi}}

\newcommand{\btheta}{\boldsymbol{\theta}}

\newcommand{\bTheta}{\boldsymbol{\Theta}}

\newcommand{\be}{\begin{equation}}
\newcommand{\ee}{\end{equation}}
\newcommand{\bea}{\begin{eqnarray}}
\newcommand{\eea}{\end{eqnarray}}
\newcommand{\bes}{\begin{eqnarray*}}
\newcommand{\ees}{\end{eqnarray*}}
\newcommand{\bi}{\begin{itemize}}
\newcommand{\ei}{\end{itemize}}
\newcommand{\bpm}{\begin{pmatrix}}
\newcommand{\epm}{\end{pmatrix}}

\newtheorem{thm}{Theorem}[section]
\newtheorem{lem}[thm]{Lemma}
\newtheorem{cor}[thm]{Corollary}
\newtheorem{prop}[thm]{Proposition}
\newtheorem{defin}[thm]{Definition}

\newcommand{\wbTheta}{{\overline{\boldsymbol{\Theta}}}}

\title{A General Framework for Association Analysis of Heterogeneous Data}
\author[1]{Gen Li}
\author[2]{Irina Gaynanova}

\affil[1]{Department of Biostatistics, Mailman School of Public Health, Columbia University}
\affil[2]{Department of Statistics, Texas A$\&$M University}

\begin{document}

\date{}

\maketitle

\begin{abstract}

Multivariate association analysis is of primary interest in many applications.
Despite the prevalence of high-dimensional and non-Gaussian data (such as count-valued or binary), most existing methods 
only apply to low-dimensional data with continuous measurements. Motivated by the Computer Audition Lab 500-song (CAL500) music annotation study, we develop a new framework for the association analysis of two sets of high-dimensional and heterogeneous (continuous/binary/count) data.
%
We model heterogeneous random variables using exponential family distributions, and exploit a structured decomposition of the underlying natural parameter matrices to identify shared and individual patterns for two data sets.  We also introduce a new measure of the strength of association, and a permutation-based procedure to test its significance. An alternating iteratively reweighted least squares algorithm is devised for model fitting, and several variants are developed to expedite computation and achieve variable selection. 
The application to the CAL500 data sheds light on the relationship between acoustic features and semantic annotations, and provides effective means for automatic music annotation and retrieval.

\end{abstract}

\section{Introduction}
With the advancement of measurement technologies, data acquisition becomes cheaper and easier.
Often, data are collected from multiple sources or different platforms on the same set of samples, which are known as multi-view or multi-modal data.
One of the main challenges associated with the analysis of multi-view data is that measurements from different sources may have heterogeneous types, such as continuous, binary, and count-valued.
For instance, the motivating Computer Audition Lab 500-song (CAL500) data \citep{turnbull2007towards} contain two sets of variables, acoustic features and semantic annotations, which are collected for 502 Western popular songs from the past 50 years.
The acoustic features characterize the audio textures of a song, and are continuous variables obtained from well-developed signal processing methods \citep[see][for example]{logan2000mel}.
The semantic annotations represent a song with a binary vector of labels over a multi-word vocabulary of semantic concepts.
The labels correspond to different genres, usages, instruments, characteristics, and vocal types.


In large music databases, it is often desired to have computers automatically generate a short description for a novel song from its acoustic features (auto-tagging), or select relevant songs based on a multi-word semantic query (music retrieval) \citep{turnbull2007towards,turnbull2008semantic,barrington2007audio,bertin2008autotagger,goto2004recent}.
The CAL500 study provides a well annotated music database to achieve these goals.
The matched acoustic features and annotation profiles facilitate the investigation of the association between the two sets of variables.
The association analysis may not only reveal how audio textures jointly affect listeners' subjective feelings, but also identify annotation patterns that can be used for music retrieval.
As a result, it may give rise to new, effective auto-tagging and retrieval methods.

One of the most popular methods for the multivariate association analysis is the canonical correlation analysis (CCA) \citep{hotelling1936relations}.
The CCA seeks linear combinations of the two sets of continuous variables with the maximal correlation. The loadings of the combinations offer insights into how the two sets of variables are related, whereas the resulting correlation is used to assess the strength of association. Furthermore, the canonical variables can be used for subsequent analyses such as regression \citep{luo2016canonical} and clustering \citep{chaudhuri2009multi}. However, the standard CCA has many limitations. On the one hand, it implicitly assumes that both sets of variables are real-valued in order to make the linear combinations interpretable. Moreover, the Gaussian assumption is used to provide a probabilistic interpretation \citep{bach2005probabilistic}. That said, the CCA is not appropriate for non-Gaussian data, such as the binary annotations in the CAL500 study.
On the other hand, the CCA suffers from overfitting for high dimensional data.
When the number of variables in either data set exceeds the sample size, the largest canonical correlation will always be one, resulting in misleading conclusions.
Several extensions have been studied in the literature to address the overfitting issue, with sparsity regularization being the most common approach \citep{witten2009penalized,chen2012efficient,chen2013sparse}.
These methods, however, are not directly applicable to non-Gaussian data.

To conduct the association analysis of the CAL500 data, we develop a new framework that accommodates  high-dimensional heterogeneous variables.
We call it the {\em Generalized Association Study} (GAS) framework.
We model heterogeneous data types (binary/count/continuous) using exponential family distributions, and exploit a structured decomposition of the underlying natural parameter matrices to capture the dependency structure between the variables.
The natural parameter matrices are specifically factorized into joint and individual structure, where the joint structure characterizes the association between the two data sets, and individual structure captures the remaining variation in each set.
The proposed framework builds upon a low-rank model, which reduces the overfitting issue for high dimensional data.
To our knowledge, this is the first attempt to generalize the multivariate association analysis to high dimensional non-Gaussian data from a frequentist perspective.
We apply the method to the CAL500 data, and explicitly characterize the dependency structure between the acoustic features and the semantic annotations.
We further use the proposed framework to devise new procedures for auto-tagging and music retrieval.
The resulting annotation performance is superior to existing methods.


%
The proposed model connects to the joint and individual variation explained (JIVE) model \citep{Lock2013} and the inter-battery factor analysis (IBFA) model \citep{tucker1958inter,browne1979maximum} under the Gaussian assumption.
\cite{Klami2010Bayesian,klami2013bayesian,virtanen2011bayesian} extended the IBFA model to non-Gaussian data under the Bayesian framework and developed Bayesian CCA methods for the association analysis.
However, the Bayesian methods require Gaussian priors for technical considerations, and are computationally prohibitive for large data.
A major difference of the proposed method is that we treat the underlying natural parameters as fixed effects and exploit a frequentist approach to estimate them without imposing any prior distribution. 
The model parameters can be efficiently estimated using generalized linear models (GLM) and the algorithm scales well to large data.
In addition, variable selection can be easily incorporated into the proposed framework to further facilitate interpretation.
{\color{black} A similar idea has been explored in the context of mixed graphical models \citep{cheng2017high,yang2014semiparametric,lee2015generalized}, which extend Gaussian graphical models to mixed data types. However, graphical models generally focus on characterizing relations between variables rather than data sets, and thus are not directly suitable for the purpose of music annotation and retrieval.}


Another unique contribution of the paper is that we introduce a new measure of the strength of association between the two heterogeneous data sets: the {\em association coefficient}.
We devise a permutation-based test which formally assesses the significance of association and provides a p-value.
We apply the methods to the CAL500 data, and identify a statistically significant, yet moderate, association between the acoustic features and the semantic annotations.
The statistical significance warrants the analysis of the dependency structure between the heterogeneous data types.
The moderate association may partially explain why auto-tagging and query-by-semantic-description are challenging problems, and no existing machine learning method provides extraordinary performance \citep{turnbull2008semantic,bertin2008autotagger}.

The rest of the paper is organized as follows.
In Section \ref{sec:model}, we introduce the model and discuss identifiability conditions under the GAS framework.
In Section \ref{sec:asso}, we describe the new association coefficient and a permutation-based hypothesis test for the significance of association.
In Section \ref{sec:alg}, we elaborate the model fitting procedure. 
In Section \ref{sec:CAL500}, we apply the proposed framework to the CAL500 data, and discuss new procedures for auto-tagging and music retrieval.
In Section \ref{sec:sim}, we conduct comprehensive simulation studies to compare our approach with existing methods.
Discussion and concluding remarks are provided in Section \ref{sec:dis}.
Proofs, technical details of the algorithm, a detailed description of the rank estimation procedure, and additional simulation results can be found in the supplementary material.

\section{Generalized Association Study Framework}\label{sec:model}
In this section, we first introduce a statistical model for characterizing the dependency structure between two non-Gaussian data sets. Then we discuss the identifiability of the proposed model.

\subsection{Model}

Let $\bX_1$ and $\bX_2$ be two data matrices of size $n\times p_1$ and $n\times p_2$, respectively, with rows being the samples (matched between the matrices) and columns being the variables.
We assume the entries of each data matrix are realizations of univariate random variables from a single-parameter exponential family distribution (e.g., Gaussian, Poisson, Bernoulli). In particular, the random variables may follow different distributions in different matrices.
The probability density function of each random variable $x$ takes the form
\[
f(x|\theta)=h(x)\exp\{x\theta-b(\theta)\},
\]
where $\theta\in\mathbb{R}$ is a natural parameter, $b(\cdot)$ is a convex cumulant function, and $h(\cdot)$ is a normalization function.
The expectation of the random variable is $\mu=b'(\theta)$.
Following the notation in the GLM framework, the canonical link function is defined as $g(\mu)={b'}^{-1}(\mu)$.
The notation for some commonly used exponential family distributions is given in Table \ref{tab:EF}.

\begin{table}[htbp]
\caption{The notation for some commonly used  exponential family distributions.}
\label{tab:EF}
\begin{center}
\begin{tabular}{lcccc}
\toprule
 & Mean $\mu$ &  Natural Parameter $\theta$ & $b(\theta)$ & $g(\mu)$ \\
\toprule
{Gaussian} & \multirow{2}{*}{$\mu$} & \multirow{2}{*}{$\mu$} & \multirow{2}{*}{${\theta^2\over2}$} & \multirow{2}{*}{$\mu$}\\
{(with unit variance)} & &&&\\
\midrule
Poisson  & $\lambda$ & $\log\lambda$ & $\exp(\theta)$ & $\log(\mu)$\\
\midrule
Bernoulli & $p$ & $\log{p\over 1-p}$ & $\log\{1+\exp(\theta)\}$ & $\log{\mu\over 1-\mu}$\\
\bottomrule
\end{tabular}
\end{center}
\end{table}

{\color{black}
Each random variable in the data matrix $\bX_k$ corresponds to a unique underlying natural parameter, and all the natural parameters form an $n\times p_k$ parameter matrix $\bTheta_k\in\mathbb{R}^{n\times {p_k}}$. 
The univariate random variables are assumed conditionally independent, given the underlying natural parameters.
The relation among the random variables is captured by the intrinsic patterns of the natural parameter matrices $\bTheta_1$ and $\bTheta_2$, which serve as the building block of the proposed model.
We remark that the conditional independence assumption given underlying natural parameters is commonly used in the literature for modeling multivariate non-Gaussian data.
See, \cite{zoh2016pcan,she2012reduced,lee2015generalized,goldsmith2015generalized}, for example.
On the one hand, univariate exponential family distributions 
are more tractable than the multivariate counterparts \citep{johnson1997discrete}.
Other than the multivariate Gaussian distribution, multivariate exponential family distributions are generally less studied and hard to use.}
On the other hand, the entry-wise natural parameters can be used to capture the statistical dependency in multivariate settings, acting similarly to a covariance matrix. For example, \cite{collins2001generalization} provided an alternative interpretation of the principal component analysis (PCA) using the low rank approximation to the natural parameter matrix.



Under the independence assumption, each entry of $\bX_k$ follows an exponential family distribution with the probability density function $f_k(\cdot)$ and the corresponding natural parameter matrix $\bTheta_k$.
To characterize the joint structure between the two data sources and the individual structure within each data source, we model $\bTheta_1$ and $\bTheta_2$ as
\be\label{GCCA}\left\{
\begin{aligned}
\bTheta_1&=\1\bmu_1^T+\bU_0\bV_1^T+\bU_1\bA_1^T\\
\bTheta_2&=\1\bmu_2^T+\bU_0\bV_2^T+\bU_2\bA_2^T
\end{aligned}\right..
\ee
Each parameter matrix is decomposed into three parts: the intercept (the first term), the joint structure (the second term) and the individual structure (the third term).
In particular, $\1$ is an length-$n$ vector of all ones and $\bmu_k$ is a length-$p_k$ {\em intercept} vector for $\bTheta_k$.
Let $r_0$ and $r_k$ denote the joint and individual ranks respectively, where $r_0\leq\min(n,p_1,p_2)$ and $r_k\leq \min(n,p_k)$. Then, $\bU_0$ is an $n\times r_0$ shared {\em score} matrix between the two parameter matrices;
$(\bV_1^T,\bV_2^T)^T$ is a $(p_1+p_2)\times r_0$ shared {\em loading} matrix, where $\bV_k$ corresponds to $\bTheta_k$ only;
$\bU_k$ and $\bA_k$ are $n\times r_k$ and $p_k\times r_k$ individual score and loading matrices for $\bTheta_k$, respectively.

The decomposition of the natural parameter matrices in \eqref{GCCA} has an equivalent form from the matrix factorization perspective.
More specifically,
\bes\label{GCCAmat}
\left(\bTheta_1,\bTheta_2\right)=\left(\1,\bU_0,\bU_1,\bU_2\right)
\bpm
\bmu_1^T & \bmu_2^T\\
\bV_1^T & \bV_2^T\\
\bA_1^T &\0\\
\0 & \bA_2^T
\epm,
\ees
where $\0$ represents any zero matrix of compatible size.
This structured decomposition sheds light on the association and specificity of the two data sources.
Loosely speaking, if the joint structure dominates the decomposition, the two parameter matrices are deemed highly associated. 
On the contrary, if the individual structure is dominant, the two data sets are less connected.
A more rigorous measure of association is given in Section \ref{sec:asso}.

\subsection{Connection to existing models}
Under the Gaussian assumption on $\bX_1$ and $\bX_2$, 
Model \eqref{GCCA} is identical to the JIVE model with two data sets \citep{Lock2013}:
{\color{black}
\begin{align*}
\bX_1 = \1\bmu_1^{T}+ \bU_0\bV_1^{T} + \bU_1\bA_1^{T}+ \bE_1,\\
\bX_2 = \1\bmu_2^{T}+ \bU_0\bV_2^{T} + \bU_2\bA_2^{T}+ \bE_2,
\end{align*}
where $\bE_1$ and $\bE_2$ are additive noise matrices. JIVE is an example of linked component models \citep{Zhou:2016fn}, where the dependency between two data sets is characterized by the presence of fixed shared latent components (i.e.g, $\bU_0$). When the shared components are absent, JIVE reduces to individual PCA models for $\bX_1$ and $\bX_2$. When the individual components are absent, JIVE reduces to a consensus PCA model \citep{Westerhuis:1998ux}. These models are closely related to the factor analysis, and the main difference is the deterministic (rather than probabilistic) treatment of latent components. If we substitute the fixed parameters $\bU_0$ and $\bU_k$ with Gaussian random variables, Model \eqref{GCCA} coincides with the IBFA model \citep{tucker1958inter,browne1979maximum}. The deterministic approach, however, allows us to interpret JIVE as a multi-view generalization of the standard PCA. While explicitly designed for modeling associations between two data sets, CCA cannot take into account individual latent components. As a result, it has been shown that linked component models often outperform CCA in the estimation of joint associations \citep{Trygg:2003fq, Jia:2010tua, Zhou:2016jj}. For further comparison between CCA and JIVE, we refer the reader to \citet{Lock2013}. 

The proposed framework extends linked component models to the exponential family distributions. Rewriting Model~\eqref{GCCA} with respect to each entry of $\bX_1$ and $\bX_2$ (denoted by $x_{1ij}$ and $x_{2ik}$) leads to
\begin{align*}
x_{1ij} &\sim f_1(\theta_{1ij}), \quad x_{2ik} \sim f_2(\theta_{2ik})\quad\mbox{with}\\
\theta_{1ij} &= \mu_{1j} + \sum_{r=1}^{r_0}u_{0ir}v_{1jr} + \sum_{l=1}^{r_1}u_{1il}a_{1jl},\\
\theta_{2ik}&= \mu_{2j} + \sum_{r=1}^{r_0}u_{0ir}v_{2kr} + \sum_{m=1}^{r_2}u_{2im}a_{2km}.
\end{align*}
where $f_1(\cdot)$ and $f_2(\cdot)$ are exponential family probability density functions associated with $\bX_1$ and $\bX_2$; and $u_{0ir}$, $u_{1il}$, $u_{2im}$, $v_{1jr}$, $v_{2kr}$, $a_{1jl}$, $a_{2km}$ are elements of $\bU_0$, $\bU_1$, $\bU_2$, $\bV_1$, $\bV_2$, $\bA_1$, and $\bA_2$, respectively. The above display reveals that $\bU_0$, $\bU_1$, $\bU_2$ can be viewed as fixed latent factors with $\bU_0$ being shared across both data sets, and $\bU_1$, $\bU_2$ being data set-specific. As such, this model is closely connected to the factor analysis in the context of generalized linear models. The factors are used to model the means of random variables through the canonical link functions rather than directly. The deterministic treatment allows us to interpret our model as a multi-view generalization of the exponential PCA \citep{collins2001generalization}, similar to JIVE as a multi-view generalization of the standard PCA.
}

\subsection{Identifiability}\label{sec:id}
To ensure the identifiability of Model \eqref{GCCA}, we consider the following regularity conditions:
%
\bi
\item The columns of the individual score matrices ($\bU_1$ and $\bU_2$) are linearly independent; the intercept ($\bmu_k$) and the columns of the joint and individual loading matrices ($\bV_k$ and $\bA_k$) corresponding to each data type are linearly independent;
\item The score matrices are column-centered (i.e., $\1^T(\bU_0,\bU_1,\bU_2)=\0$), and the column space of the joint score matrix is orthogonal to that of the individual score matrices  (i.e., $\bU_0^T(\bU_1,\bU_2)=\0$);
\item Each score matrix has orthogonal columns, and each loading matrix has orthonormal columns (i.e., $\bV_1^T\bV_1+\bV_2^T\bV_2=\bI$, $\bA_1^T\bA_1=\bI$ and $\bA_2^T\bA_2=\bI$, where $\bI$ is an identity matrix of compatible size).
\ei
The first condition ensures that the joint and individual ranks are correctly specified.
The second condition orthogonalizes the intercept, the joint and the individual patterns.
The last condition rules out the arbitrary rotation and rescaling of each decomposition, if the column norms of respective score matrices are distinct (this is almost always true in practice).
{\color{black}We remark that the orthonormality condition for the concatenated joint loadings in $(\bV_1^T,\bV_2^T)^T$ is more general than separate orthonormality conditions for $\bV_1$ and $\bV_2$, and is beneficial for modeling data with different scales and structures.}
Under the above conditions, Model \eqref{GCCA} is uniquely defined up to trivial column reordering and sign switches.
The rigorous proof of the model identifiability partially attributes to the Theorem 1.1 in the supplementary material of \cite{Lock2013}. For completeness, we restate the theorem under our framework:
\begin{prop}\label{thm:unique}
Let
\begin{equation*}
\left\{
\begin{aligned}
\bTheta_1&=\bJ_1+\bB_1,\\
\bTheta_2&=\bJ_2+\bB_2,
\end{aligned}
\right.
\end{equation*}
$\bJ=(\bJ_1,\bJ_2)$ and $\bB=(\bB_1,\bB_2)$, where $\rank(\bJ)=r_0$ and $\rank(\bB_k)=r_k$ for $k=1,2$.
Suppose the model ranks are correctly specified, i.e., $\rank(\bB)=r_1+r_2$ and $\rank(\bTheta_k)=r_0+r_k$ for $k=1,2$.
There exists a unique parameter set  $\{\bJ_1,\bJ_2,\bB_1,\bB_2\}$ satisfying $\bJ^T\bB=\0$.
\end{prop}

In Model \eqref{GCCA}, we have $\bJ_k=\1\bmu_k^T+\bU_0\bV_k^T$ and $\bB_k=\bU_k\bA_k^T$ ($k=1,2$).
Our first identifiability condition is equivalent to the rank prerequisite in the proposition 2.1.
The second condition guarantees $\bJ^T\bB=\0$.
Hence the joint and individual patterns of our model are uniquely defined.
Furthermore, our last identifiability condition is the standard condition that guarantees the uniqueness of the singular value decomposition (SVD) of a matrix \citep{golub2012matrix}.

%
%

\section{Association Coefficient and Permutation Test}\label{sec:asso}
\subsection{Association Coefficient}
Model \eqref{GCCA} specifies the joint and individual structure of the natural parameter matrices underlying the two data sets.
The relative weights of the joint structure can be used to measure the strength of association between the two data sources.
Intuitively, if the joint structure dominates the individual structure, the latent generating schemes of the two data sets are coherent. Consequently, the two data sources are deemed highly associated.
On the contrary, if the joint signal is weak, each data set roughly follows an independent EPCA generative model~\citep{collins2001generalization}, and hence the two data sources are unrelated. To formalize this idea, we define an association coefficient between the two data sets as follows.
\begin{defin}\label{d:rho}
Let $\bX_1\in\real^{n\times p_1}$ and $\bX_2 \in \real^{n\times p_2}$ be two data sets with $n$ matched samples, and assume $\bX_k$ ($k=1,2$) follows an exponential family distribution with the entrywise underlying natural parameter matrix $\bTheta_k$. Let $\wbTheta_k$ be the column centered $\bTheta_k$. The \textbf{association coefficient} between $\bX_1$ and $\bX_2$ is defined as
\be\label{ac}
    \rho(\bX_1,\bX_2)=\frac{\|\wbTheta_1^T\wbTheta_2\|_\star}{\|\wbTheta_1\|_\mathbb{F}\|\wbTheta_2\|_\mathbb{F}},
\ee
where $\|\cdot\|_\star$ and $\|\cdot\|_\mathbb{F}$ represent the nuclear norm and Frobenius norm of a matrix, respectively.
In particular, under Model \eqref{GCCA} with the identifiability conditions, the association coefficient has the expression $$\rho(\bX_1,\bX_2)={\|\bV_1\bU_0^T\bU_0\bV_2^T+\bA_1\bU_1^T\bU_2\bA_2^T\|_\star\over\|\bU_0\bV_1^T+\bU_1\bA_1^T\|_\mathbb{F}\|\bU_0\bV_2^T+\bU_2\bA_2^T\|_\mathbb{F}}.$$
\end{defin}

The definition of the association coefficient \eqref{ac} only depends on the natural parameter matrix underlying each data set. It does not rely on our model assumption. Thus it is applicable in a broad context.
Furthermore, the association coefficient satisfies the following properties.
The proof can be found in Section A of the supplementary material.

\begin{prop}
\bi
\item[(i)] The association coefficient $\rho(\bX_1,\bX_2)$ is bounded between 0 and 1.
\item[(ii)] $\rho(\bX_1,\bX_2)=0$ if and only if the column spaces of $\wbTheta_1$ and $\wbTheta_2$ are mutually orthogonal.
\item[(iii)] $\rho(\bX_1,\bX_2)=1$ if $\wbTheta_1$ and $\wbTheta_2$ have the same left singular vectors and proportional singular values.
\ei
\end{prop}


The first property puts the association coefficient on scale, making it similar to the conventional notion of correlation.
A smaller value means weaker association, and vice versa.
The second and third properties establish the conditions for ``no association" and ``perfect association", respectively.
We remark that the second property provides a necessary and sufficient condition for $\rho(\bX_1,\bX_2)=0$, while the third property only provides a sufficient condition for $\rho(\bX_1,\bX_2)=1$.
In the context of Model \eqref{GCCA}, we have the following corollary.
\begin{cor}
Suppose Model \eqref{GCCA} has correctly specified ranks and satisfies the identifiability conditions. Then,
\bi
\item[(i)] $\rho(\bX_1,\bX_2)=0$, if and only if $\bU_0=\0$ and $\bU_1^T\bU_2=\0$;
\item[(ii)] $\rho(\bX_1,\bX_2)=1$, if $\bU_1=\0$, $\bU_2=\0$, $\bV_1^T\bV_1=c\bI$ and $\bV_2^T\bV_2=(1-c)\bI$ for some constant $0<c<1$.
\ei
\end{cor}
Conceptually, the association coefficient is zero when the joint structure is void and the individual patterns are mutually orthogonal in both data sets.
Perhaps less obvious are the conditions for the two data sets to have the association coefficient exactly equal to one.
Not only the individual structure does not exist, but the columns of $\bV_1$ (and $\bV_2$) must be mutually orthogonal with the same norm.
It turns out the additional rigor is necessary.
It reduces the risk of overestimating the association under model misspecification.
See Section A of the supplementary material for some concrete examples.


\subsection{Permutation Test}\label{sec:test}
To formally assess the statistical significance of the association between $\bX_1$ and $\bX_2$, we consider the following hypothesis test:
\bes
\mbox{H}_0: \rho(\bX_1,\bX_2)=0 \ \mbox{ vs } \ \mbox{H}_1: \rho(\bX_1,\bX_2)>0.
\ees
We use the sample version of the association coefficient $\rho(\bX_1,\bX_2)$ as the test statistic, and exploit a permutation-based testing procedure.

More specifically, assume $\wbTheta_1$ and $\wbTheta_2$ are estimated from data (see Section \ref{sec:alg} for parameter estimation).
The original test statistic, denoted by $\rho_0$, can be obtained from \eqref{ac}.
Now we describe the permutation procedure.
Let $\bP_\pi$ be an $n\times n$ permutation matrix with the random permutation $\pi: \{1,\cdots,n\}\mapsto\{1,\cdots, n\}$.
We keep $\bX_1$ fixed and permute the rows of $\bX_2$ based on $\pi$.
As a result, the association between the two data sets is removed while the respective structure is reserved.
The corresponding association coefficient for the permuted data, denoted by $\rho_\pi$, is a random sample under the null hypothesis.
{\color{black}Because the natural parameters are defined individually and permuted along with $\bX_2$, the column centered natural parameter matrix for $\bP_\pi\bX_2$ is $\bP_\pi\wbTheta_2$.}
Thus, we directly obtain the expression of $\rho_\pi$ as
\[
\rho_\pi=\frac{\|\wbTheta_1^T\bP_\pi\wbTheta_2\|_\star}{\|\wbTheta_1\|_\mathbb{F}\|\bP_\pi\wbTheta_2\|_\mathbb{F}}=\frac{\|\wbTheta_1^T\bP_\pi\wbTheta_2\|_\star}{\|\wbTheta_1\|_\mathbb{F}\|\wbTheta_2\|_\mathbb{F}}.
\]
We repeat the permutation procedure multiple times and get a sampling distribution of the association coefficient under the null.
Consequently,  the empirical p-value is calculated as the proportion of permuted values greater than or equal to the original test statistic $\rho_0$.
A small p-value warrants further investigation on the dependency structure between the two data sets.

\section{Model Fitting Algorithm}\label{sec:alg}
In this section, we elaborate an alternating algorithm to estimate the parameters in Model \eqref{GCCA}. We show that the model fitting procedure can be formulated as a collection of GLM fitting problems.
We also discuss how to incorporate variable selection into the framework via a regularization approach.
{\color{black}When fitting the model, we assume the joint and individual ranks are fixed.
We briefly introduce how to select the ranks at the end of this section. A more detailed data-driven rank selection approach is presented in Section D of the supplementary material.}

\subsection{Alternating Iteratively Reweighted Least Square}
The model parameters in \eqref{GCCA} consist of the intercept $\bmu_k$, the joint score $\bU_0$, the individual score $\bU_k$, the joint loading $\bV_k$, and the individual loading $\bA_k$ ($k=1,2$).
To estimate the parameters, we maximize the joint log likelihood of the observed data $\bX_1$ and $\bX_2$, denoted by $\ell(\bX_1,\bX_2|\bTheta_1,\bTheta_2)$.
Under the independence assumption, the joint log likelihood can be written as the summation of the individual log likelihoods for each value. Namely, we have
\be\label{lik}
\ell(\bX_1,\bX_2|\bTheta_1,\bTheta_2)=\sum_{i=1}^n\sum_{j=1}^{p_1}\ell_1(x_{1,ij}|\theta_{1,ij})+\sum_{i=1}^n\sum_{j=1}^{p_2}\ell_2(x_{2,ij}|\theta_{2,ij}),
\ee
where $\bX_k=(x_{k,ij})$ and $\bTheta_k=(\theta_{k,ij})$, and $\ell_k$ is the log likelihood function for the $k$th distribution ($k=1,2$).
In particular, $\bTheta_1$ and $\bTheta_2$ have the structured decomposition in \eqref{GCCA}.
We estimate the parameters in a block-wise coordinate descent fashion: we alternate the estimation between the joint and the individual structure, and between the scores and the loadings (with the intercepts), until convergence.


More specifically, we first fix the joint structure $\{\bU_0, \bV_1, \bV_2\}$, and estimate the individual structure for each data set.
Since the first term in \eqref{lik} only involves $\{\bmu_1,\bU_1,\bA_1\}$, and the second term only involves $\{\bmu_2,\bU_2,\bA_2\}$, the parameter estimation is separable.
We focus on the first term, and the second term can be updated similarly.
We first fix $\bmu_1$ and $\bA_1$ to estimate $\bU_1$.
Let $\bu_{k,(i)}$ be the column vector of the $i$th row of $\bU_k$ ($k=0,1,2$).
The column vector of the $i$th row of $\bTheta_1$, denoted by $\btheta_{1,(i)}$, can be expressed as
\bes
\btheta_{1,(i)}=\bmu_1+\bV_1\bu_{0,(i)}+\bA_1\bu_{1,(i)},
\ees
where everything is fixed except for $\bu_{1,(i)}$.
Noticing that  the $i$th row of $\bX_1$ (i.e., $\bx_{1,(i)}$) and $\btheta_{1,(i)}$ satisfy
\bes\label{GLM}
\Ex(\bx_{1,(i)})=b_1'\left(\btheta_{1,(i)}\right),
\ees
we exactly obtain a GLM with the canonical link.
Namely, $\bx_{1,(i)}$ is a generalized response vector; $\bA_1$ is a $p_1\times r_1$ predictor matrix; $\bmu_1+\bV_1\bu_{0,(i)}$ is an offset; $\bu_{1,(i)}$ is a coefficient vector.
The estimate of $\bu_{1,(i)}$ can be obtained via an iteratively reweighted least squares (IRLS) algorithm \citep{mccullagh1989generalized}.
Furthermore, different rows of $\bU_1$ can be estimated in parallel.
Overall, the estimation of $\bU_1$ is formulated as $n$ parallel GLM fitting problems.
Once $\bU_1$ is estimated, we fix $\bU_1$ and formulate the estimation of $\bmu_1$ and $\bA_1$ as $p_1$ GLMs in a similar fashion.
Consequently, we update the estimate of the individual structure.


Now we estimate the joint structure with fixed individual structure.
When the joint score $\bU_0$ is fixed, the estimation of $\{\bmu_1,\bV_1\}$ and $\{\bmu_2,\bV_2\}$ resembles the estimation of the individual counterparts.
With fixed $\{\bmu_1,\bmu_2,\bV_1,\bV_2\}$, the estimation of $\bU_0$ is slightly different, because it is shared by two data types with different distributions.
Let $\btheta_{0,(i)}=(\btheta_{1,(i)}^T,\btheta_{2,(i)}^T)^T$ be a column vector concatenating the column vectors of the $i$th rows of $\bTheta_1$ and $\bTheta_2$.
Then we have
\bes
\btheta_{0,(i)}=\left(\bmu_1^T+\bu^T_{1,(i)}\bA_1^T\ , \ \bmu_2^T+\bu^T_{2,(i)}\bA_2^T\right)^T +\bV_0\bu_{0,(i)},
\ees
where $\bV_0=(\bV_1^T,\bV_2^T)^T$ is the concatenated joint loading matrix. Notice that
\bes
\Ex(\bx_{1,(i)})=b_1'(\btheta_{1,(i)}),\quad \Ex(\bx_{2,(i)})=b_2'(\btheta_{2,(i)}).
\ees
The formula corresponds to a non-standard GLM where the response consists of observations from different distributions, and different link functions are used correspondingly.
Following the standard GLM model fitting algorithm verbatim, we obtain a slightly modified version of the IRLS algorithm to address this problem. More details can be found in Section B of the supplementary material.

The separately estimated parameters, denoted by $\{\widehat{\bmu_1},\widehat{\bmu_2},\widehat{\bU_0},\widehat{\bU_1},\widehat{\bU_2},\widehat{\bV_1},\\
\widehat{\bV_2},\widehat{\bA_1},\widehat{\bA_2}\}$, may not satisfy the identifiability conditions in Section \ref{sec:id}. In order to find an equivalent set of parameters satisfying the conditions, we conduct the following normalization procedure after each iteration.
We first project the columns of the individual scores $\bU_1$ and $\bU_2$ to the orthogonal complement of the column space of $(\1,\bU_0)$.
The obtained individual score matrices are denoted by $\bU_1^\star$ and $\bU_2^\star$, which are column centered and orthogonal to the columns in $\bU_0$.
The new individual patterns are $\bU_1^\star\widehat{\bA_1}^T$ and $\bU_2^\star\widehat{\bA_2}^T$ accordingly.
To rule out arbitrary rotations and scale changes, we apply the SVD to each individual structure, and let the left singular vectors to absorb the singular values.
As a result, we have
\[
\widetilde{\bU_1}\widetilde{\bA_1}^T=\bU_1^\star\widehat{\bA_1}^T,\quad \widetilde{\bU_2}\widetilde{\bA_2}^T=\bU_2^\star\widehat{\bA_2}^T,
\]
where $\{\widetilde{\bU_1},\widetilde{\bU_2},\widetilde{\bA_1},\widetilde{\bA_2}\}$ satisfies the identifiability conditions.
Next, we add the remaining individual structure to the joint structure, and obtain the new joint structure as
\[
\left(\1\widehat{\bmu_1}^T+\widehat{\bU_0}\widehat{\bV_1}^T+\widehat{\bU_1}\widehat{\bA_1}^T-\widetilde{\bU_1}\widetilde{\bA_1}^T ,\
\1\widehat{\bmu_2}^T+\widehat{\bU_0}\widehat{\bV_2}^T+\widehat{\bU_2}\widehat{\bA_2}^T-\widetilde{\bU_2}\widetilde{\bA_2}^T\right).
\]
Denote the new column mean vector as $\left(\widetilde{\bmu_1}^T,\widetilde{\bmu_2}^T\right)^T$, and center each column of the above joint structure.
Subsequently, we apply SVD to the column-centered joint structure and obtain the new joint score $\widetilde{\bU_0}$ and joint loading $\left(\widetilde{\bV_1}^T,\widetilde{\bV_2}^T\right)^T$.
As a result, the new parameter set $\{\widetilde{\bmu_1},\widetilde{\bmu_2},\widetilde{\bU_0},\widetilde{\bU_1},\\
\widetilde{\bU_2},\widetilde{\bV_1},\widetilde{\bV_2},\widetilde{\bA_1},\widetilde{\bA_2}\}$ satisfies all the conditions, and provides the same likelihood value as the original parameter set.

In summary, we devise an alternating algorithm to estimate the model parameters. Each iteration is formulated as a set of GLMs, fitted by the IRLS algorithm. A step-by-step summary is provided in Algorithm \ref{alg1}.
Because the likelihood value in \eqref{lik} is nondecreasing in each optimization step, and remains constant in the normalization step, the algorithm is guaranteed to converge.
More formally, we have the following proposition.

\begin{prop}\label{prop1}
In each iteration of Algorithm \ref{alg1}, the log likelihood \eqref{lik} is monotonically nondecreasing. If the likelihood function is bounded, the estimates always converge to some stationary point (including infinity).
\end{prop}

\begin{algorithm}[h]
\caption{The Alternating IRLS Algorithm for Fitting Model \eqref{GCCA}}\label{alg1}
\begin{algorithmic}
\State Initialize $\{\bmu_1,\bmu_2,\bU_0,\bU_1,\bU_2,\bV_1,\bV_2,\bA_1,\bA_2\}$;
\While {The likelihood \eqref{lik} has not reached convergence}
\bi
\item Fix the joint structure $\{\bU_0,\bV_1,\bV_2\}$
    \bi
    \item Fix $\{\bmu_1,\bA_1\}$, and estimate each row of $\bU_1$ via parallel GLM
    \item Fix $\bU_1$, and estimate each row of $(\bmu_1,\bA_1)$ via parallel GLM
    \item Fix $\{\bmu_2,\bA_2\}$, and estimate each row of $\bU_2$ via parallel GLM
    \item Fix $\bU_2$, and estimate each row of $(\bmu_2,\bA_2)$ via parallel GLM
    \ei
\item Fix the individual structure $\{\bU_1,\bU_2,\bA_1,\bA_2\}$
    \bi
    \item Fix $\bU_0$, and estimate each row of $(\bmu_1,\bV_1)$ via parallel GLM
    \item Fix $\bU_0$, and estimate each row of $(\bmu_2,\bV_2)$ via parallel GLM
    \item Fix $\{\bmu_1,\bmu_2,\bV_1,\bV_2\}$, and estimate each row of $\bU_0$ via a modified IRLS algorithm in parallel
    \ei
\item Normalize the estimated parameters to retrieve the identifiability conditions
\ei
\EndWhile
\end{algorithmic}
\end{algorithm}

Since the overall algorithm is iterative, we further substitute the IRLS algorithm with a one-step approximation with warm start to enhance computational efficiency.
A detailed description is provided in Section C of the supplementary material.
In our numerical studies, we observe that the one-step approximation algorithm almost always converges to the same values as the full algorithm, but is several fold faster (see Section \ref{sec:sim}).

%

\subsection{Variable Selection}\label{vs}
In practice, it is often desirable to incorporate variable selection into parameter estimation to facilitate interpretation, which is especially relevant when the number of variables is high.
Various regularization frameworks and sparsity methods have been extensively studied in the literature. See \cite{hastie2015statistical} and references therein.

Since Model \eqref{GCCA} is primarily used to investigate the association between the two data sets, it is of great interest to perform variable selection when estimating the joint structure.
In particular, sparse $\bV_1$ and $\bV_2$ facilitate model interpretability.
The variables corresponding to non-zero joint loading entries can be used to interpret the association between the two data sources.

In order to achieve variable selection in the estimation, we modify the normalization step in each iteration of the model fitting algorithm.
In particular, we substitute the SVD of the centered joint structure with the FIT-SSVD method developed by \cite{yang2014sparse}.
The FIT-SSVD method provides sparse estimation of the singular vectors via soft or hard thresholding, while maintaining the orthogonality among the vectors.
By default, an asymptotic threshold is used to automatically determine the sparsity level for each data set.
Consequently, the method is directly embedded into our algorithm to generate sparse estimates.
The final estimates of $\bV_1$ and $\bV_2$ may be sparse, and the estimated parameters satisfy the identifiability conditions.
We remark that FIT-SSVD can be applied to the individual structure as well if desired.

{\color{black}
\subsection{Rank Estimation}
In order to estimate $(r_0,r_1,r_2)$, we adopt a two-step procedure.
The first step is to estimate the ranks of the column centered natural parameter matrices for $\bX_1$, $\bX_2$, and $(\bX_1,\bX_2)$.
In order to achieve that, we devise an $N$-fold cross validation approach. 
The idea is as follows: we first randomly split the entries of a data matrix into $N$ folds; then we withhold one fold of data and use the rest to estimate natural parameter matrices with different ranks via an alternating algorithm; finally we calculate the cross validation score corresponding to each rank by taking the average of squared Pearson residuals of the withheld data. The candidate rank with the smallest score will be selected.
We remark that the approach can flexibly accommodate a data matrix from a single non-Gaussian distribution, or a data matrix consisting of mixed variables from multiple distributions (e.g., $(\bX_1,\bX_2)$).
We apply the approach to $\bX_1$, $\bX_2$, and $(\bX_1,\bX_2)$, respectively, and obtain  the estimated ranks $r_1^\star$, $r_2^\star$, and $r_0^\star$.

In the second step, we solve a system of linear equations to estimate $(r_0,r_1,r_2)$.
From Model \eqref{GCCA} and the identifiability conditions, we have the following relations: $r_0^\star=r_0+r_1+r_2$, $r_1^\star=r_0+r_1$, and $r_2^\star=r_0+r_2$.
Therefore, the estimate of $(r_0,r_1,r_2)$ is obtained by
\bes
r_0=r_1^\star+r_2^\star-r_0^\star, \quad r_1=r_0^\star-r_1^\star, \quad r_2=r_0^\star-r_1^\star.
\ees
A more detailed description of the two-step rank estimation procedure and comprehensive numerical studies can be found in Section D of the supplementary material.
}


\section{CAL500 Music Annotation}\label{sec:CAL500}

In this section, we analyze the CAL500 data. 
The data are publicly available at the {\sf Mulan} database \citep{tsoumakas2011mulan}.
The CAL500 data consist of 502 popular songs.
The audio signal of each song has been analyzed via signal processing methods, and converted to 68 continuous features.
The features are generally partitioned into 5 categories: spectral centroid, spectral flux, spectral roll-off, zero crossings, and Mel-Frequency Cepstral Coefficients (MFCC), measuring different aspects of an audio profile.
In addition, each song has been manually annotated by multiple listeners.
There are 174 total annotations, related to the emotion (36 variables), genre (47), usage (15), instrument (33), characteristic (27) and vocal type (16) of a song.
Each song has been assigned a binary sequence of annotations based on the responses from listeners.
A more detailed description can be found in \cite{turnbull2007towards}.

There are two data sets with matched samples but distinct data types in CAL500.
The primary goal is to understand the association between the two sets of variables (i.e., acoustic features and semantic annotations), and leverage the information to achieve automatic annotation and music retrieval. The proposed GAS framework is suitable for the association analysis.
In the following, we first elaborate the model fitting procedure with the CAL500 data, and then describe the annotation and retrieval performance.

\subsection{Model Fitting}\label{sec:fit}
Let $\bX_1$ denote the continuous acoustic features and $\bX_2$ denote the binary semantic annotations.
We have $n=502$, $p_1=68$ and $p_2=174$.
Each column of $\bX_1$ has been centered and normalized to have unit standard deviation.
Furthermore, we exploit SVD to estimate the standard deviation of the random noise in $\bX_1$ as $\sigma$, and scale the entire data matrix by $1/\sigma$ so that the noise has unit variance.
Consequently, we model the preprocessed data $\bX_1$ by Gaussian distributions with the structured mean matrix $\bTheta_1$ in Model \eqref{GCCA} and unit variance.
We model the binary data matrix $\bX_2$ by Bernoulli distributions with the structured natural parameter matrix $\bTheta_2$ in Model \eqref{GCCA}.

We use a data-driven approach  to estimate the model ranks to be $\widehat{r_0}=3$, $\widehat{r_1}=3$ and $\widehat{r_2}=2$.
A detailed description is provided in Section D of the supplementary material.
Subsequently, we fit Model \eqref{GCCA} to the CAL500 data with the estimated ranks.
We exploit the one-step approximated version of the algorithm without sparsity.
The algorithm converges at high accuracy within 300 iterations, taking less than 3 minutes on a desktop (Intel i5 CPU (3.3GHz) with 8Gb RAM).

We calculate the association coefficient \eqref{ac} based on the estimated parameters and get $\rho=0.265$.
The coefficient indicates a moderate association between the acoustic features and the semantic annotations.
Furthermore, we conduct the permutation-based association test (with 1000 permutations) as  described in Section \ref{sec:test}. The permuted statistics roughly follow a Gaussian distribution (see Figure \ref{fig:test}). The empirical p-value of the test is 0.
Namely, the  association between the acoustic features and the semantic annotations is highly statistically significant.

\begin{figure}[htbp]
  \centering
  \includegraphics[width=2.8in]{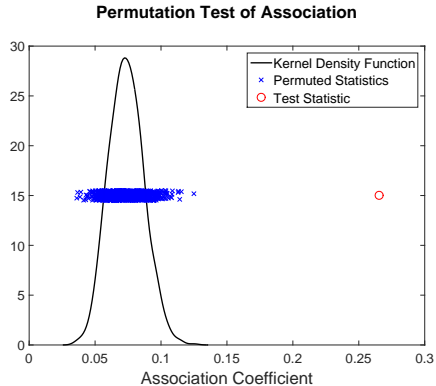}
  \caption{Permutation-based association test for the CAL500 data. The kernel density is estimated from 1000 permuted association coefficients. The original test statistic (red circle) and the permuted statistics (blue cross) are shown in the plot with random jitters on the y axis for the ease of visualization.}\label{fig:test}
\end{figure}

We further investigate the three joint loading vectors.
For each loading, we sort the variables in each data source based on the loading values from large to small.
In the first joint loading vector, annotations corresponding to the largest positive values include emotions such as ``Soft", ``Calming" and ``Loving", and Usage such as "Romancing."
Annotations corresponding to the largest negative values include emotions such as ``Aggressive" and ``Angry", and genres such as ``Metal Hard Rock."
Namely, the first loading primarily captures the emotion of a song.
The corresponding top acoustic features are the MFCCs and the zero crossings, which are known to measure the noisiness of audio signals.
The second joint loading mainly characterizes the attitude of a song (e.g., ``Cheerful" vs ``Not Cheerful", ``Danceable" vs ``Not Danceable").
Music genres such as ``R\&B", ``Soul" and ``Swing" also have large positive loading values on the cheerful side, which is quite intuitive.
The corresponding top acoustic features include the MFCCs and the zero crossings, as well as the spectral centroid, which measures the `brightness' of the music texture.
The third joint loading captures more subtle patterns.
For annotations, genres such as ``Jazz" and ``Bebop" and characteristics such as ``Changing Energy Level" and ``Positive Feelings" have large positive values, while genres ``Country", ``Roots Rock", ``Hip-Hop" and ``Rap" have large negative values.
The top acoustic features are dominated by the MFCCs.

\subsection{Automatic Annotation}
Under the GAS framework, we propose the following procedure to automatically annotate a new song based on its acoustic features.
Suppose we have all the model parameters, $\{\bmu_k, \bU_0, \bV_k, \bU_k, \bA_k; k=1,2\}$, estimated from a training data set. Given a new song with the acoustic feature vector $\bx_1^\star\in \real^{p_1}$, we first estimate the corresponding joint and individual scores $\left({\bu_0^\star}^T,{\bu_1^\star}^T\right)^T$ by regressing $\bx_1^\star-\bmu_1$ on $(\bV_1,\bA_1)$.
Next, we extract the joint score $\bu_0^\star$ and obtain an estimate of the annotation natural parameters via $\btheta_2^\star=\bmu_2+\bV_2\bu_0^\star$. Finally, we convert the natural parameters to probabilities via the entry-wise logistic transformation $\bpi^\star=\exp(\btheta_2^\star)/(1+\exp(\btheta_2^\star))$. Consequently, each entry of $\bpi^\star$ provides the probability of the song having the corresponding annotation. In other words, $\bpi^\star$ is the induced annotation profile of the song.
In practice, one could preset a threshold, and output the semantic descriptions in the vocabulary with probabilities greater than the threshold as the annotation of the song.

To compare the proposed method with existing auto-tagging approaches, we conduct a 10-fold cross validation study on the CAL500 data, similar to that in \cite{turnbull2008semantic}.
For simplicity, we select 500 out of the 502 songs in the data, and randomly partition them into 10 blocks, each having 50 songs.
In each run, we use 452 songs as the training set, and test on the remaining 50 songs.
To be consistent with \cite{turnbull2008semantic}, we annotate each test song with exactly ten annotations (the top ten annotations with the largest probabilities in $\bpi^\star$ according to our method).

The annotation performance is assessed by the mean per-word precision and recall.
More specifically, for each annotation, let $t_{GT}$ be the number of songs in the test set that have the annotation in the human-generated ``ground truth"; let $t_{A}$ be the number of songs that are annotated with the tag by a method; let $t_{TP}$ be the number of ``true positives" that have the tag both in the ground truth and in the automatic annotation prediction. The per-word precision is defined as $t_{TP}/t_{A}$, and the per-word recall is $t_{TP}/t_{GT}$.
The mean per-word precision and recall are calculated by averaging the values across different tags in each cross validation run.
Annotations with undefined precision or recall are omitted when calculating the mean.

We compare the proposed method with the MixHier method \citep{turnbull2008semantic} and the Autotagger method \citep{bertin2008autotagger}.
We also consider two baseline methods, a ``Random" lower bound and an empirical upper bound (denoted by ``UpperBnd"), for precision and recall, as discussed in \cite{turnbull2008semantic}.
Loosely speaking, the Random approach randomly selects ten annotations for each test song based on the observed tag frequencies, and mimics a random guessing procedure.
The UpperBnd approach serves as the best-case scenario.
It uses the ground truth to annotate test songs, and randomly adds or removes tags to meet the ten-annotation requirement.  
The mean and standard deviation of the mean per-word precision and recall for different methods from the 10-fold cross validation are presented in Table \ref{tab:ann}. 

\begin{table}[htbp]
\caption{The CAL500 automatic annotation results. The mean and standard deviation (in parenthesis) for mean per-word precision (``Precision") and mean per-word recall (``Recall") across 10 cross validation runs are presented. 
The best results are bold-faced}\label{tab:ann}
\begin{center}
\begin{tabular}{l|c|c}
\toprule
Method & Precision &  Recall \\
\toprule
Random & 0.144 (0.004) & 0.064 (0.002)\\
UpperBnd & 0.712 (0.007) & 0.375 (0.006) \\ 
\midrule
MixHier & 0.265 (0.007) & {\bf 0.158} (0.006) \\
Autotagger & 0.312 (0.060) & 0.153 (0.015) \\
Proposed & {\bf0.438} (0.051) & 0.078 (0.007)\\
\bottomrule
\end{tabular}
\end{center}
\end{table}

Overall, all three methods are significantly better than random guessing, but considerably worse than the empirical upper bounds.
The suboptimal results may be justified by the moderate association between the acoustic features and the semantic annotations (see Section \ref{sec:fit}).
Namely, only a moderate amount of information in the annotations can be explained by the existing acoustic features. Thus, to further improve the automatic annotation performance, more comprehensive characterization of the audio profile may be needed.

Although a good balance of precision and recall is desired, it has been argued that precision is more relevant for recommender systems \citep{herlocker2000explaining}.
The proposed method has the best precision among all three methods.
Thus, it may provide an effective approach for auto-tagging.
The relatively low recall may be due to the small number of predicted annotations (i.e., 10) per song.
{\color{black}We further increase the number of words used to characterize a song to 20, and redo the analysis. As a result, we get a recall rate of 0.154 with standard deviation 0.015, which is comparable to the competing methods, and a precision rate of 0.330 with standard deviation 0.036, which is still superior to the competing methods.}
We further investigate the complete annotation profile of each song using the proposed method.
Figure \ref{fig:anno} shows four randomly selected examples. The top and bottom bars in each plot correspond to the estimated and true annotation profiles. We particularly order the annotations for visualization convenience.
The proposed method produces sensible results.
It captures the majority of the true annotations with large probabilities, and has much richer patterns.
Whether the additional annotations with high probabilities are false positives or missing tags due to the well-known ``human bias" issue in music tagging \citep{ellis2002quest} remains an open question.

\begin{figure}[htbp]
  \centering
  \includegraphics[width=1.7in]{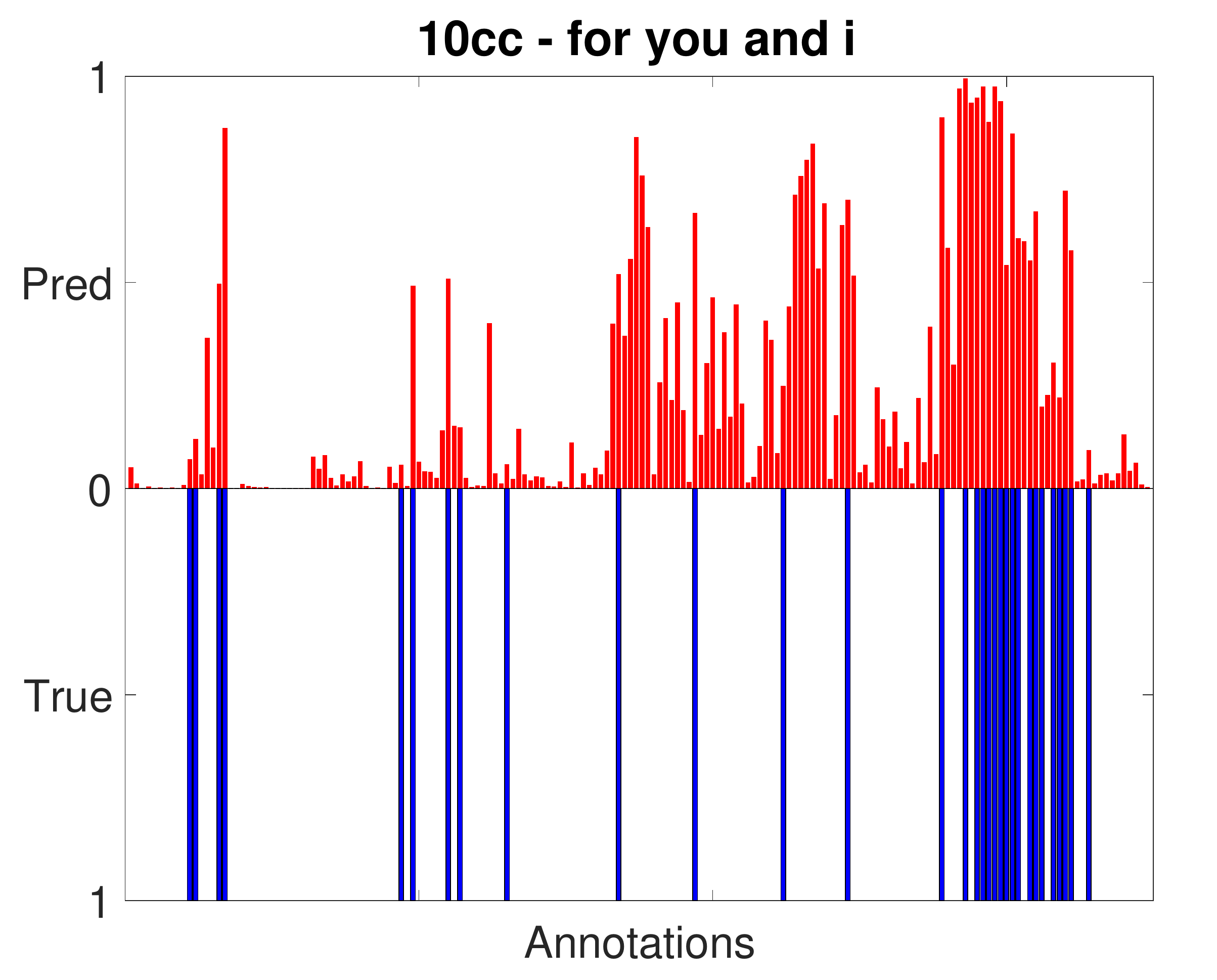}
  \includegraphics[width=1.7in]{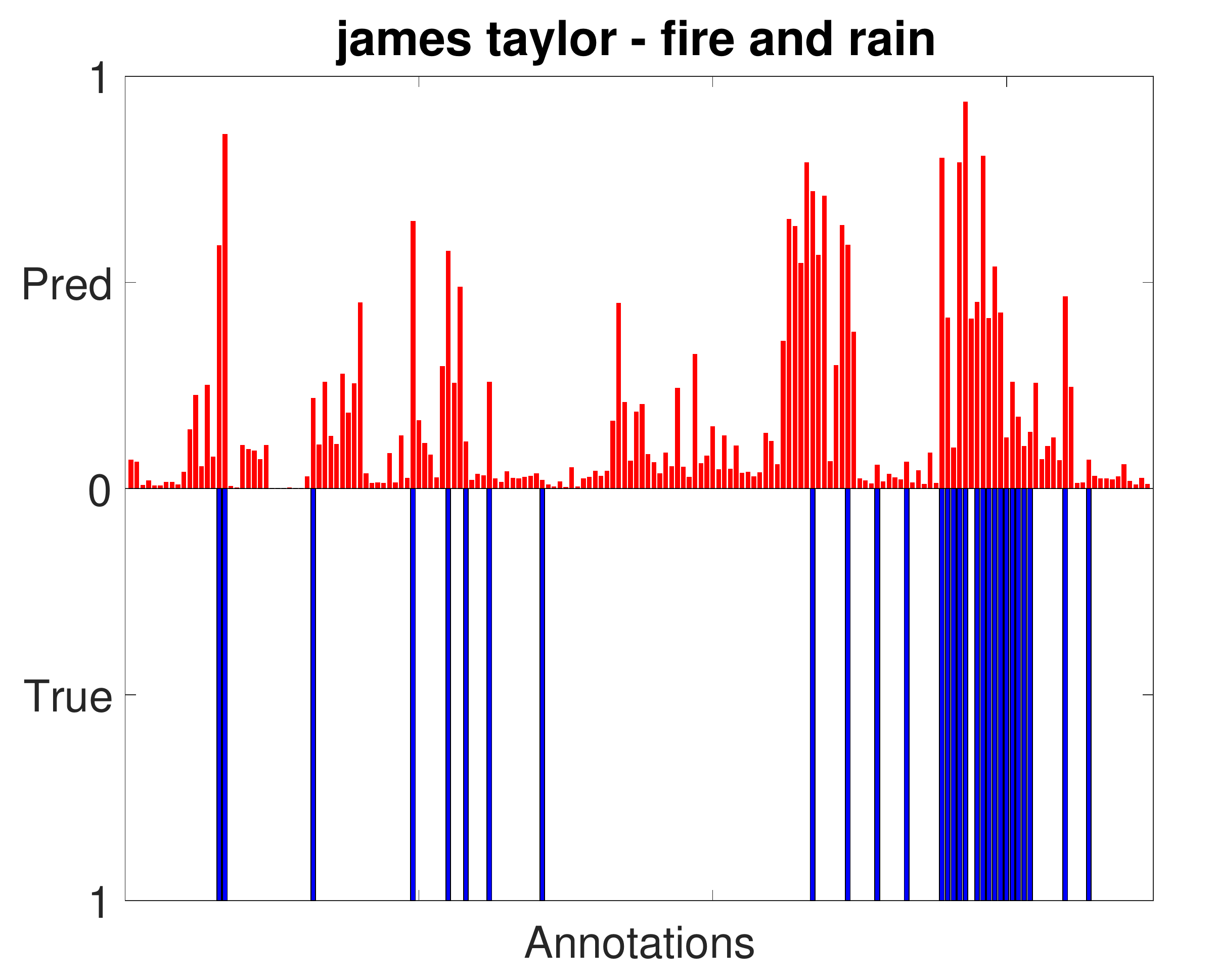}\\
  \includegraphics[width=1.7in]{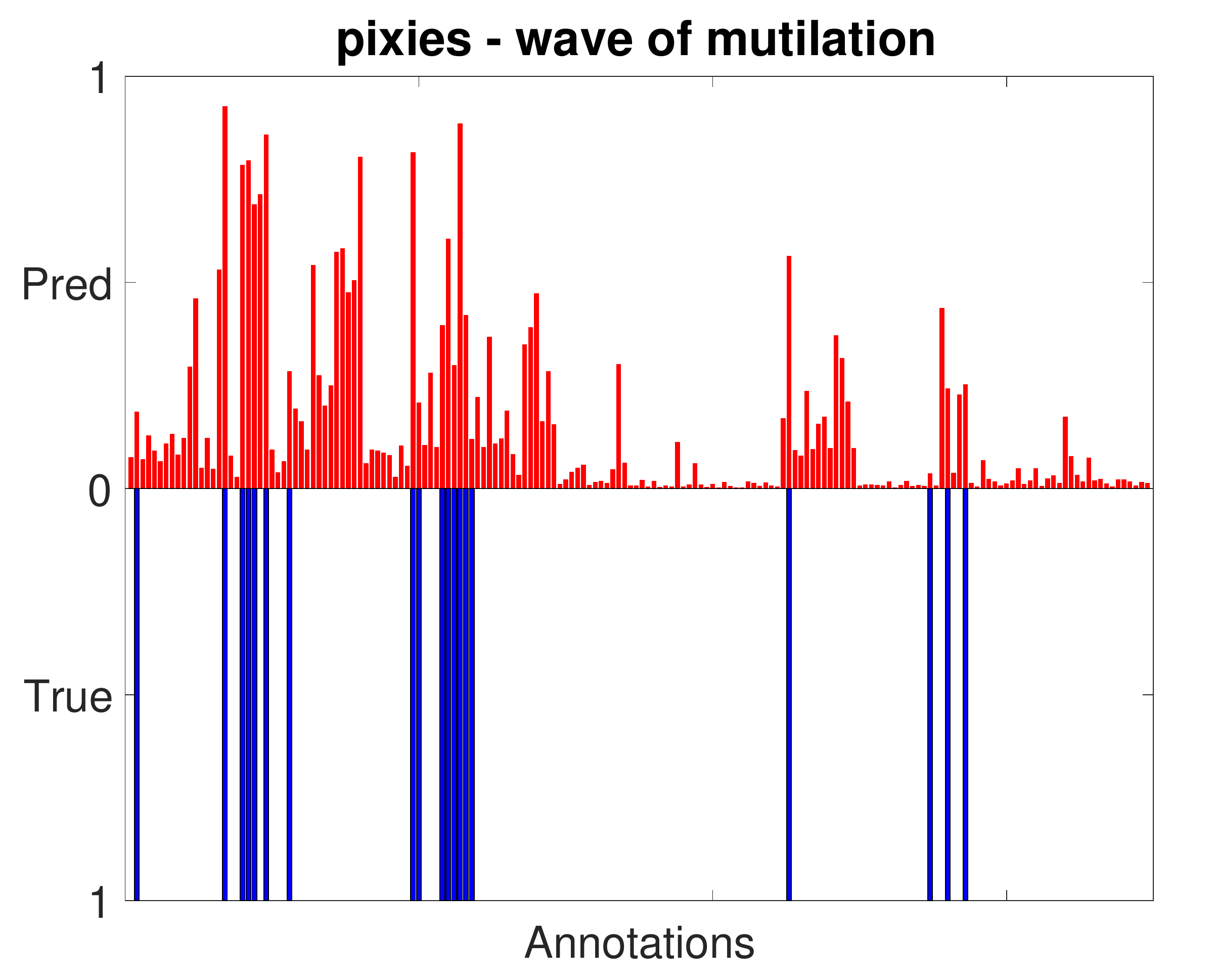}  \includegraphics[width=1.7in]{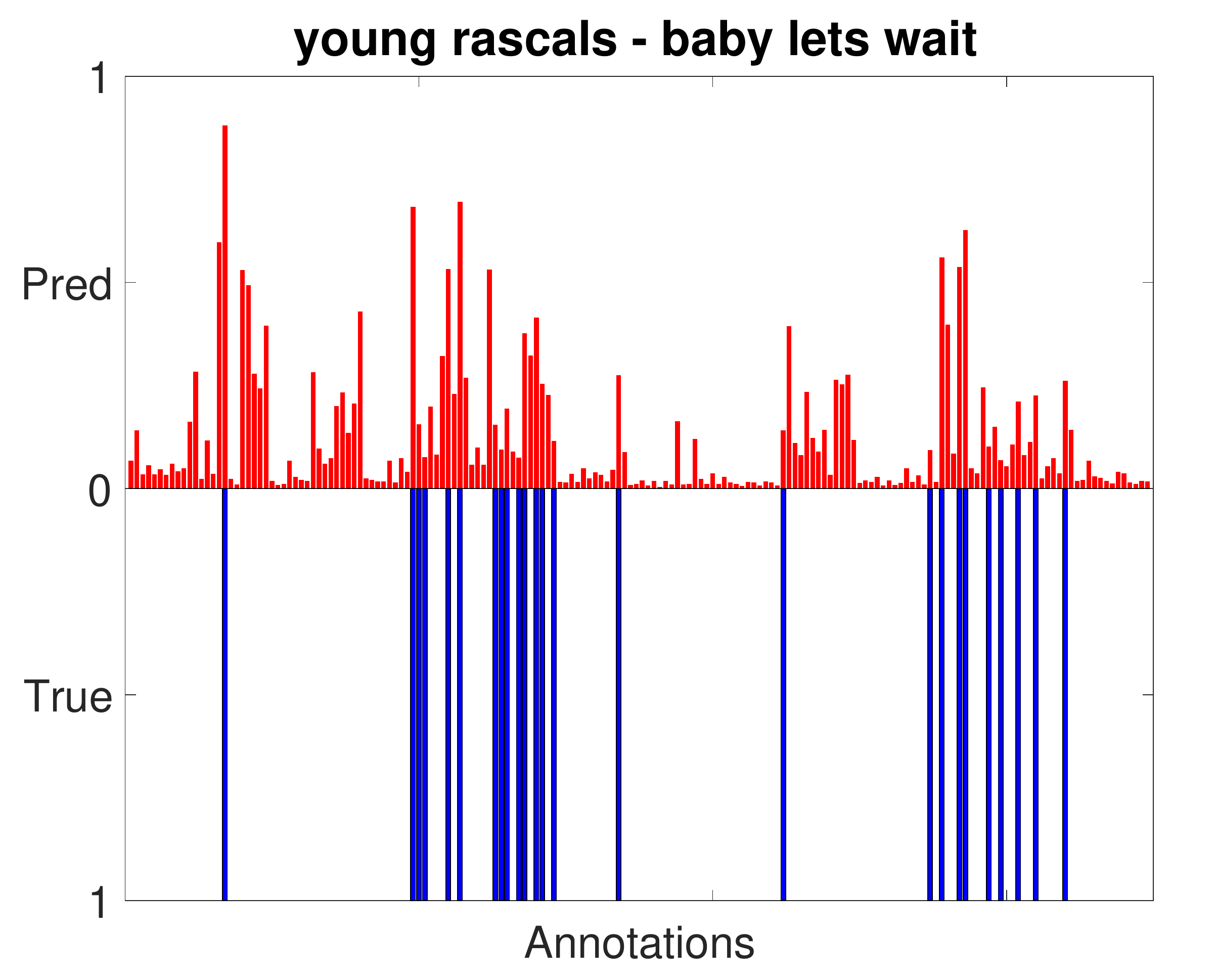}
  \caption{The CAL500 automatic annotation results. Each plot corresponds to a song. In each plot, the top red bars provide the predicted annotation profile; the lower blue bars correspond to the true annotations.The annotations are ordered for visualization convenience.}\label{fig:anno}
\end{figure}

\subsection{Music Retrieval}
We also investigate music retrieval using the proposed framework.
We remark that finding songs based on a small set of annotations is relatively easy. One could simply filter the songs in the database by the given tags, and output those satisfying all the requirements. Thus it is not our primary interest here.
Instead, we focus on retrieving songs according to a more complicated query consisting of multiple tags.

Similar to automatic annotation, we propose the following procedure for music retrieval based on a given annotation list.
Suppose the model parameters in \eqref{GCCA} have been estimated.
For any given query , we first convert it to a binary vector $\bx_2^\star$ using the semantic annotation library. Then, we regress $\bx_2^\star$ on $(\bV_2,\bA_2)$ using a logistic regression with offset $\bmu_2$, and obtain the estimate of the joint and individual scores $\bu_0^\star$ and $\bu_2^\star$.
Next, we calculate the Mahalanobis distances between the estimated score vector $\left({\bu_0^\star}^T,{\bu_2^\star}^T\right)^T$ and the score vectors corresponding to the songs in the database.
The covariance matrix used in the Mahalanobis distance is estimated from the model parameter $(\bU_0,\bU_2)$.
Finally, we sort the distances in an ascending order.
As a result, we obtain an ordered list with highest recommendation on the top.

To validate the procedure, we apply it to the CAL500 data.
We use the annotation profile of each song as a query.
For each query, we record the ranking of the reference song (also contained in the database) in the output recommendation list.
Figure \ref{fig:retrieval} shows the histogram of the rankings across the 502 requests.
As desired, most of the time, the reference song is among the top of the recommendation list.
Perhaps what's more interesting are the top choices other than the reference song in each request.
They are the most similar songs to the reference song in the database according to the  annotation query.
For instance, for the song ``For You and I" by {\sf 10cc}, the top recommendations include ``God Bless the Child" by {\sf Billie Holiday}, ``Suzanne" by {\sf Leonard Cohen} and ``Postcard Blues" by {\sf Cowboy Junkies}.
Without ``ground truth" of the true rankings, however, further validation of the music retrieval performance remains an open question \citep{ellis2002quest}.

\begin{figure}[htbp]
  \centering
  \includegraphics[width=2.5in]{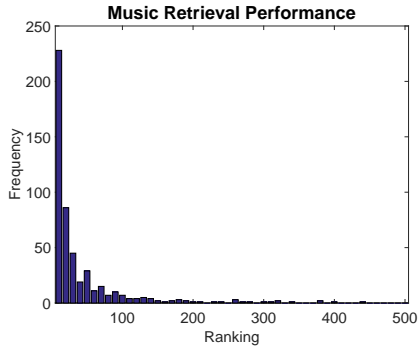}
  \caption{The CAL500 music retrieval result. The histogram of the reference song rankings across different music retrieval requests. }\label{fig:retrieval}
\end{figure}

\section{Simulation Study}\label{sec:sim}
In this section, we conduct comprehensive simulation studies to compare the proposed method with existing ones. We consider several versions of the method: the double-iterative version (denoted by ``iter-GAS") as described in Algorithm \ref{alg1}, the one-step version (``GAS") as described in Section C of the supplementary material, and the one-step with sparsity version (``sGAS") as described in Section \ref{vs}. In addition, we also consider an ad hoc competing method derived from EPCA \citep{collins2001generalization} and JIVE \citep{Lock2013}, where we first estimate a low-rank individual natural parameter matrix for each data set via EPCA, and then apply JIVE to the two estimated matrices.
We denote the ad hoc approach by EPCA-JIVE.

We generate data from Model \eqref{GCCA}, and apply different methods to estimate model parameters.
To avoid complication, we set the joint and individual ranks for the GAS methods to be the true ranks.
{\color{black} In Section G of the supplementary material, we further investigate the effect of rank misspecification on the performance.}
For the EPCA-JIVE method, in the EPCA step, we set the rank of each individual natural parameter matrix to be a large number (much larger than the true rank) in order to avoid information loss.
In particular, for Gaussian data, we use the full rank, or equivalently, the original data.
In the JIVE step, we use the true joint and individual ranks.
The assessment of the rank estimation procedure is conducted separately in Section D.3 of the supplementary material.

\subsection{Setting}
We set the sample size to be $n=200$, and the dimensions of both data sets to be $p_1=p_2=120$.
The joint and individual ranks of the column-centered natural parameter matrices are $r_0=r_1=r_2=2$.
The scores in $(\bU_0,\bU_1,\bU_2)$ are filled with random numbers generated from a uniform distribution between $-0.5$ to $0.5$ (i.e., $Unif(-0.5,0.5)$), and normalized via the Gram-Schmidt process to have orthonormal columns.
We particularly consider 4 settings of the natural parameters, and perform 100 simulation runs for each with the same underlying parameters.
\bi
\item {\bf Setting 1 (Gaussian-Gaussian)}: The joint loadings $(\bV_1^T,\bV_2^T)^T$ are generated in a similar way to the scores: filled with uniform random numbers and normalized to have orthonormal columns. The respective individual loadings $\bA_1$ and $\bA_2$ are similarly generated to satisfy the identifiability conditions. We set the singular values of the joint structure to be $(180,140)$, and of the individual structure to be $(120,100)$ and $(100,80)$. All singular values are absorbed into the scores. The intercepts $\bmu_1$ and $\bmu_2$ are filled with $Unif(-0.5,0.5)$.
\item {\bf Setting 2 (Gaussian-Bernoulli)}: The loadings are generated similarly to Setting 1, except that $\bV_1$ (Gaussian) and $\bV_2$ (Bernoulli) are initially filled with $Unif(-0.5,0.5)$ and  $Unif(-1,1)$ before the normalization. The singular values of the joint structure are $(240,220)$ and those for the individual structure are $(90,80)$ and $(200,180)$. The intercept is filled with $Unif(-0.5,0.5)$.
\item {\bf Setting 3 (Gaussian-Poisson)}: The loadings are generated similarly to Setting 1, except that $\bV_1$ (Gaussian) and $\bV_2$ (Poisson) are initially filled with $Unif(-0.5,0.5)$ and $Unif(-0.25,0.25)$. The singular values are $(80,40)$ (joint), $(60,40)$ (Gaussian individual), and $(20,16)$ (Poisson individual). The intercept terms $\bmu_1$ and $\bmu_2$ are filled with $Unif(-0.5,0.5)$ and $Unif(2,3)$ respectively.
\item {\bf Setting 4 (Bernoulli-Poisson)}: The loadings are generated similarly to Setting 1, except that $\bV_1$ (Bernoulli) and $\bV_2$ (Poisson) are initially filled with $Unif(-5,5)$ and $Unif(-0.5,0.5)$ respectively. The singular values are $(180,140)$ (joint), $(200,160)$ (Bernoulli individual), and $(12,10)$ (Poisson individual). The intercept terms $\bmu_1$ and $\bmu_2$ are filled with $Unif(-0.5,0.5)$ and $Unif(2,3)$ respectively.
\ei
Once the natural parameters are fixed, the observed data are generated independently from corresponding distributions.
In particular, for Gaussian random numbers, we set the variance to be one.

We remark that for Bernoulli distribution, the scale of the natural parameters needs to be relatively large in order to have a detectable signal. Hence we purposely increase the corresponding singular values and the relative loading scales for the Bernoulli distribution in {\bf Setting 2 and 4}. For Poisson distribution, due to the asymmetry of the canonical link function, the natural parameters are typically skewed towards positive values. To mimic reality, we set the intercept term for the Poisson distribution to be positive in {\bf Setting 3 and 4}.


We also consider the settings where the joint loadings are sparse. As the results for sparse settings are qualitatively similar to the results in dense settings, we refer the reader to Section F of supplementary material. 


\subsection{Result}
We compare GAS, iter-GAS, and EPCA-JIVE on the non-sparse simulation settings. Each method is applied to the simulated data to estimate the model parameters.
We evaluate the loading estimation accuracy by the maximum principal angle \citep{bjorck1973numerical} between the subspaces spanned by the estimated and the true loading vectors.
We consider the angles for the joint loadings $\angle(\bV_0,\widehat{\bV_0})$ (where $\bV_0=\left(\bV_1^T,\bV_2^T\right)^T$) and for separate individual loadings $\angle(\bA_k,\widehat{\bA_k})$ ($k=1,2$), respectively.
We assess the estimation accuracy of different model parameters (i.e., the intercept, the joint, and the individual structure) by the Frobenius norm of the difference between the true and the estimated values.
In particular, we calculate the following quantities ($k=1,2$):
\bes
Norm_{avg}&=&\|\bmu_k-\widehat{\bmu_k}\|_\mathbb{F},\\ Norm_{jnt}&=&\|\bU_0\bV_k^T-\widehat{\bU_0}\widehat{\bV_k}^T\|_\mathbb{F},\\
Norm_{ind}&=&\|\bU_k\bA_k^T-\widehat{\bU_k}\widehat{\bA_k}^T\|_\mathbb{F},
\ees
where $\|\cdot\|_\mathbb{F}$ represents the Frobenius norm.
Moreover, we also calculate the Frobenius loss of the overall natural parameter estimates $Norm_{\bTheta}=\|\bTheta_k-\widehat{\bTheta_k}\|_\mathbb{F}$.
In addition, we compare the model fitting times for different methods.
The results are summarized in Table \ref{tab:1}.

\begin{sidewaystable}[htbp]
\caption{Simulation results based on 100 simulation runs in each setting. The median and median absolute deviation (in parenthesis) of each criterion for different methods across different settings are presented. For each method,  $Norm_{avg}$, $Norm_{jnt}$, $Norm_{ind}$, $Norm_{\bTheta}$ and $\angle(\bA_k,\widehat{\bA_k})$ are evaluated and compared per data set; $\angle(\bV_0,\widehat{\bV_0})$ is evaluated across two data sets.  The best results are highlighted in bold.}\label{tab:1}
{\scriptsize
\begin{center}
\begin{tabular}{lcccccccc}
 & &  \multicolumn{2}{c}{GAS} & \multicolumn{2}{c}{iter-GAS} & \multicolumn{2}{c}{EPCA-JIVE} \\
 & &  Data 1 & Data 2 &   Data 1 & Data 2 &  Data 1 & Data 2 \\
\toprule
 \multirow{7}{*}{{\bf Setting 1}}
       & $\|\bmu_k-\widehat{\bmu_k}\|_\mathbb{F}$ &   {\bf 0.78}(0.03) &  {\bf 0.77}(0.04) &  {\bf 0.78}(0.03) & {\bf 0.77}(0.04) & {\bf 0.78}(0.03)&  {\bf 0.77}(0.04)\\
       & $\|\bU_0\bV_k^T-\widehat{\bU_0}\widehat{\bV_k}^T\|_\mathbb{F}$ &   {\bf 21.32}(0.43) &  {\bf 21.15}(0.41) &  {\bf 21.32}(0.43) & {\bf 21.15}(0.41) &21.33(0.42) &  {\bf 21.15}(0.41)\\
       & $\|\bU_k\bA_k^T-\widehat{\bU_k}\widehat{\bA_k}^T\|_\mathbb{F}$ &   {\bf25.39}(0.51) &  {\bf 25.65}(0.53) &  {\bf25.39}(0.51) & {\bf25.65}(0.53) &{\bf25.39}(0.51) &  {\bf25.65}(0.53)\\
       & $\|\bTheta_k-\widehat{\bTheta_k}\|_\mathbb{F}$    & {\bf34.61}(0.39)  & {\bf34.58}(0.49) & {\bf34.61}(0.39) & {\bf34.58}(0.49)&{\bf 34.61}(0.40) & {\bf34.58}(0.49) \\
       & $\angle(\bA_k,\widehat{\bA_k})$    &{\bf6.27}(0.27)  &  {\bf 7.96}(0.30)  &  {\bf6.27}(0.27) & {\bf7.96}(0.30)& {\bf6.27}(0.26) & {\bf7.96}(0.30) \\
       & $\angle(\bV_0,\widehat{\bV_0})$    & \multicolumn{2}{c}{{\bf6.36}(0.20)}  & \multicolumn{2}{c}{{\bf6.36}(0.20)}  &\multicolumn{2}{c}{{\bf6.36}(0.20)} \\
       & Time (sec)   & \multicolumn{2}{c}{10.04(0.82)}  & \multicolumn{2}{c}{44.78(3.27)}  &\multicolumn{2}{c}{{\bf0.51}(0.01)} \\
\midrule
 \multirow{6}{*}{{\bf Setting 2}}
       & $\|\bmu_k-\widehat{\bmu_k}\|_\mathbb{F}$ &   {\bf0.78}(0.04) &  {2.54}(0.10) &  {\bf0.78}(0.03) & {\bf1.96}(0.10) & {\bf0.78}(0.04)&  2.59(0.10)\\
       & $\|\bU_0\bV_k^T-\widehat{\bU_0}\widehat{\bV_k}^T\|_\mathbb{F}$ &   {\bf23.69}(0.45) &  {\bf 89.36}(5.63) &  42.79(0.56) & 128.98(1.00) & 25.15(0.48) &  185.51(1.07)\\
       & $\|\bU_k\bA_k^T-\widehat{\bU_k}\widehat{\bA_k}^T\|_\mathbb{F}$ &   {\bf26.00}(0.40) &  {\bf 110.89}(5.30) &  26.01(0.45) & 133.88(1.04) & 26.11(0.44) &  174.32(1.04)\\
       & $\|\bTheta_k-\widehat{\bTheta_k}\|_\mathbb{F}$    & {\bf36.08}(0.45)  & {\bf146.86}(7.47) & {50.80}(0.45) & 187.77(0.96)&{ 37.09}(0.48) & 257.07(1.14) \\
       & $\angle(\bA_k,\widehat{\bA_k})$    &{\bf8.18}(0.40)  &  {14.47}(0.69)  &  8.20(0.38) & {\bf13.95}(0.60)& 8.24(0.38) & 22.03(0.99) \\
       & $\angle(\bV_0,\widehat{\bV_0})$    & \multicolumn{2}{c}{12.96(0.79)}  & \multicolumn{2}{c}{{\bf12.70}(0.40)}  &\multicolumn{2}{c}{29.46(0.43)} \\
       & Time (sec)    & \multicolumn{2}{c}{{\bf10.94}(1.36)}  & \multicolumn{2}{c}{55.13(6.39)}  &\multicolumn{2}{c}{43.21(3.71)} \\
\midrule
 \multirow{6}{*}{{\bf Setting 3}}
       & $\|\bmu_k-\widehat{\bmu_k}\|_\mathbb{F}$ &   {\bf0.77}(0.03) &  {\bf 0.23}(0.01) &  {\bf0.77}(0.03) & {\bf0.23}(0.01) & {\bf0.77}(0.03)&  0.25(0.01)\\
       & $\|\bU_0\bV_k^T-\widehat{\bU_0}\widehat{\bV_k}^T\|_\mathbb{F}$ &   {\bf18.65}(0.49) &  {\bf 6.68}(0.14) &  {\bf18.65}(0.49) & 6.69(0.14) & 76.32(4.29) &  22.16(3.58)\\
       & $\|\bU_k\bA_k^T-\widehat{\bU_k}\widehat{\bA_k}^T\|_\mathbb{F}$ &   {\bf26.31(0.53)} &  {\bf 7.16}(0.16) &  {\bf26.31}(0.53) & {\bf7.16}(0.16) & 76.63(4.00) &  28.22(3.04)\\
       & $\|\bTheta_k-\widehat{\bTheta_k}\|_\mathbb{F}$    & 33.98(0.45)  & {\bf10.15}(0.13) & {\bf33.97}(0.45) & {\bf10.15}(0.13)&{ 37.86}(0.46) & 18.93(0.13) \\
       & $\angle(\bA_k,\widehat{\bA_k})$    &{\bf15.96}(0.77)  &  {\bf 11.49}(0.55)  &  {\bf15.96}(0.77) & {\bf11.49}(0.55)& 84.31(4.17) & 88.51(1.00) \\
       & $\angle(\bV_0,\widehat{\bV_0})$    & \multicolumn{2}{c}{{\bf16.28}(0.60)}  & \multicolumn{2}{c}{{\bf16.28}(0.60)}  &\multicolumn{2}{c}{85.68(3.21)} \\
       & Time (sec)    & \multicolumn{2}{c}{{\bf23.10}(1.28)}  & \multicolumn{2}{c}{111.32(6.58)}  &\multicolumn{2}{c}{54.15(6.59)} \\
\midrule
 \multirow{6}{*}{{\bf Setting 4}}
       & $\|\bmu_k-\widehat{\bmu_k}\|_\mathbb{F}$ &   2.36(0.12) &  {\bf 0.23}(0.01) &  {\bf1.87}(0.08) & {\bf0.23}(0.01) & 2.48(0.07)&  0.24(0.01)\\
       & $\|\bU_0\bV_k^T-\widehat{\bU_0}\widehat{\bV_k}^T\|_\mathbb{F}$ &   {\bf82.99}(4.23) &  {\bf 6.17}(0.11) &  101.71(1.16) & 7.81(0.17) & 203.54(3.13) &  16.59(0.89)\\
       & $\|\bU_k\bA_k^T-\widehat{\bU_k}\widehat{\bA_k}^T\|_\mathbb{F}$ &   {\bf106.96}(5.51) &  {\bf 7.50}(0.15) &  119.11(1.09) & 7.54(0.15) &233.41(0.77) &  20.11(0.88)\\
       & $\|\bTheta_k-\widehat{\bTheta_k}\|_\mathbb{F}$    & {\bf138.99}(5.22)  & {\bf10.17}(0.14) & {157.89}(1.22) & 11.27(0.15)&{ 218.95}(1.21) & 13.96(0.14) \\
       & $\angle(\bA_k,\widehat{\bA_k})$    &14.37(0.84)  &  {\bf 18.88}(0.94)  &  {\bf13.29}(0.74) & 18.97(0.92)& 86.86(1.96) & 88.57(0.90) \\
       & $\angle(\bV_0,\widehat{\bV_0})$    & \multicolumn{2}{c}{15.39(1.02)}  & \multicolumn{2}{c}{{\bf14.98}(0.78)}  &\multicolumn{2}{c}{87.59(1.64)} \\
       & Time (sec)   & \multicolumn{2}{c}{{\bf7.42}(0.63)}  & \multicolumn{2}{c}{35.53(3.18)}  &\multicolumn{2}{c}{81.13(5.01)} \\


\bottomrule
\end{tabular}
\end{center}
}
\end{sidewaystable}

We observe that under {\bf Setting 1} where the two data sets are both Gaussian, all three methods have very similar performances.
In particular, GAS and iter-GAS are identical because the IRLS algorithm degenerates to the ordinary least squares under the Gaussian assumption.
Model \eqref{GCCA} coincides with the JIVE model in this setting, and thus GAS provides an alternative way of fitting the JIVE model.
In {\bf Setting 2} where the distributions are Gaussian and Bernoulli, the GAS method is generally the best (except for the mean structure and loading estimation in the second data set). For Bernoulli distributions, sometimes the maximum likelihood of EPCA and iter-GAS is reached at infinity, posing a convergence issue to both methods. The same issue has been pointed out in \cite{collins2001generalization}.  As a remedy, we introduce a small ridge penalty to the GLM likelihood functions. This allows the algorithm to converge to a finite value. However, the resulting estimates are biased and shrunk towards zero. See Section E of the supplementary material for more details. We remark that the one-step approximation algorithm is more robust against the convergence issue, and typically does not require such a penalty. Consequently, the estimates are more accurate.
In {\bf Setting 3} where the distributions are Gaussian and Poisson, GAS and iter-GAS have similar results, both outperforming the EPCA-JIVE method.
In {\bf Setting 4} where the distributions are Bernoulli and Poisson, again, GAS is generally among the best in almost all aspects, followed by iter-GAS. Both provide more accurate estimates than EPCA-JIVE.
In terms of the computational cost, the one-step GAS method is always more efficient than the iterative GAS method. Both outperform the ad hoc approach except for the Gaussian case.

{\color{black} As suggested by a referee, we also investigate the performance of the GAS method in high dimensional settings. We focus on {\bf Setting 3} and consider two variants with dimensions $p_1=p_2=200$ and $p_1=p_2=300$, respectively. We keep the signal-to-noise ratio constant as the dimensions increase. Analysis results show that the estimation accuracy further improves with increasing dimensions due to the ``blessing of dimensionality" \citep{li2017embracing},  demonstrating the efficacy of the GAS method in high dimensional settings. More details can be found in Section G of the supplementary material.

In addition, we also study the proposed method in the situation where ranks are misspecified. Results show that the estimation of underlying natural parameter matrices, loading subspaces, and association coefficients is very robust against rank misspecification. More details can be found in Section H of the supplementary material.}


\section{Discussion}\label{sec:dis}
In this paper, we develop a generalized association study framework for estimating the dependency structure and testing the significance of association between two heterogeneous data sets.
We analyze the CAL500 music annotation data with the proposed method, and identify a statistically significant but moderate association between the acoustic features and the semantic annotations.
By leveraging the information in both data sets, we develop new auto-tagging and music retrieval methods that with superior precision performance over existing approaches.
As such, they may serve as useful tools for recommender systems.



There are a few interesting directions for future research.
{\color{black}First, for the music annotation study, it is compelling to investigate what additional audio features may significantly enhance the association with the semantic annotations and improve the auto-tagging performance.
Second, from a methodological point of view, the proposed framework may be extended to over-dispersed distributions and/or to more than two data sets. How to simultaneously estimate dispersion parameters is an open question.}
Third, the application of the proposed methods to other areas such as multi-omics studies is open and promising.

\section*{Acknowledgement}
The authors would like to thank the Computer Audition Laboratory at the University of California, San Diego, for generating the CAL500 data.
GL's research was partially supported by the Calderone Junior Faculty Award by the Mailman School of Public Health at Columbia University.


\newpage
\appendix
\begin{center}
{\Large\bf Supplementary Materials for \\``A General Framework for Association Analysis of Heterogeneous Data'' by Gen Li and Irina Gaynanova}
\end{center}

\renewcommand\theequation{S.\arabic{equation}}
\setcounter{equation}{0}
\newcommand{\newcaption}{%
  \setlength{\abovecaptionskip}{-20pt}%
  \setlength{\belowcaptionskip}{0pt}%
  \caption}
\renewcommand{\thefigure}{S\arabic{figure}}
\renewcommand{\thetable}{S\arabic{table}}

\section{Proof of Proposition 3.2 and Extensions}
In this section, we first prove Proposition 3.2 in the main paper, and then prove Corollary 3.3.
Afterwards, we provide a couple of examples to demonstrate the intuition behind the proposed association coefficient.
\subsection{Proof of Proposition 3.2}
{\bf We first prove part (i).}
From the definition, it is straightforward to see that $\rho(\bX_1,\bX_2)\ge 0$.
What remains to be shown is $\|\wbTheta_1^T\wbTheta_2\|_\star\leq \|\wbTheta_1\|_\mathbb{F}\|\wbTheta_2\|_\mathbb{F}$.
This follows directly from the following lemma.
\begin{lem}
Let $\bX$ be an $n\times p$ matrix in $\real$. Then
\[
\|\bX\|_\star=\min_{\bA,\bB: \bX=\bA\bB}\|\bA\|_\mathbb{F}\|\bB\|_\mathbb{F}.
\]
\end{lem}
\begin{proof}
  Let $\bX=\bU\bD\bV^T$ be the singular value decomposition (SVD) of the rank-$r$ ($r\leq \min(n,p)$) matrix $\bX$, where $\bU\in\real^{n\times r}$ and $\bV\in\real^{p\times r}$ are the left and right singular matrices with orthonormal columns respectively, and $\bD$ is an $r\times r$ diagonal matrix with positive non-increasing singular values on the diagonal.
  For any real matrices $\bA$ and $\bB$ such that $\bX=\bA\bB$, we have $\bU\bD\bV^T=\bA\bB$, and correspondingly $\bD=\bU^T\bA\bB\bV$. Subsequently,
  \[
  \|\bX\|_\star=\tr(\bD)=\tr(\bU^T\bA\bB\bV).
  \]
  Let $\vect(\bX)$ denote the vectorization of $\bX$ along the columns. According to the Cauchy-Schwarz inequality, we have
  $$\tr(\bU^T\bA\bB\bV)=\langle\vect(\bA^T\bU),\vect(\bB\bV)\rangle\leq \|\bU^T\bA\|_\mathbb{F}\|\bB\bV\|_\mathbb{F}.$$
  Furthermore, $\|\bU^T\bA\|^2_\mathbb{F}=\tr(\bU^T\bA\bA^T\bU)=\tr((\bU,\widetilde{\bU})^T\bA\bA^T(\bU,\widetilde{\bU}))-\tr(\widetilde{\bU}^T\bA\bA^T\widetilde{\bU})=\|\bA\|^2_\mathbb{F}-\|\widetilde{\bU}^T\bA\|^2_\mathbb{F}$, where $\widetilde{\bU}\in\real^{n\times(n-r)}$ contains a set of basis of the orthogonal complement to the column space of $\bU$. Namely, $$\|\bU^T\bA\|_\mathbb{F}\leq \|\bA\|_\mathbb{F},$$ and similarly we can show $\|\bB\bV\|_\mathbb{F}\leq\|\bB\|_\mathbb{F}$.
  Combining all the results together, we have
  \[
  \|\bX\|_\star\leq\|\bA\|_\mathbb{F}\|\bB\|_\mathbb{F}.
  \]

  Let $\bA=\bU\bD^{1\over 2}$ and $\bB=\bD^{1\over2}\bV^T$. It is easy to see that $\bX=\bA\bB$ and $\|\bA\|_\mathbb{F}=\|\bB\|_\mathbb{F}=\sqrt{\tr(\bD)}$, and hence $\|\bX\|_\star=\|\bA\|_\mathbb{F}\|\bB\|_\mathbb{F}$. This concludes the proof.
\end{proof}

{\bf Next, we prove part (ii).}
The association coefficient is zero if and only if the numerator is zero.
Namely, $\rho(\bX_1,\bX_2)=0$ if and only if $\|\wbTheta_1^T\wbTheta_2\|_\star=0$.
Furthermore,  $\|\wbTheta_1^T\wbTheta_2\|_\star=0$ if and only if all the singular values of $\wbTheta_1^T\wbTheta_2$ are zero, i.e., $\wbTheta_1^T\wbTheta_2=\0$.
Thus, the necessary and sufficient condition of $\rho(\bX_1,\bX_2)=0$ is $\col(\wbTheta_1)$ orthogonal to $\col(\wbTheta_2)$, where $col(\cdot)$ represents the column space of a matrix.

\vskip.1in
{\bf Finally, we prove part (iii).}
Let $\wbTheta_k=\bU_k\bD_k\bV_k^T$ be the SVD of $\wbTheta_k$ ($k=1,2$).
If $\bU_1=\bU_2$ and $\bD_1=c\bD_2$ for some constant $c>0$, we have
\[
\wbTheta_1^T\wbTheta_2=\bV_1\bD_1\bU_1^T\bU_2\bD_2\bV_2^T=\bV_1\bD_1\bD_2\bV_2^T=c\bV_1\bD_2^2\bV_2^T.
\]
Because $\bV_1^T\bV_1=\bI$, $\bV_2^T\bV_2=\bI$, and $\bD_2^2$ is diagonal,  we know $c\bV_1\bD_2^2\bV_2^T$ is the SVD of $\wbTheta_1^T\wbTheta_2$, and hence
\[
\|\wbTheta_1^T\wbTheta_2\|_\star=\tr(c\bD^2_2)=c\|\bD_2\|_\mathbb{F}^2.
\]
In addition, we have
\[
\|\wbTheta_k\|_\mathbb{F}=\|\bU_k\bD_k\bV^T_k\|_\mathbb{F}=\|\bD_k\|_\mathbb{F},\quad k=1,2.
\]
Namely, $\|\wbTheta_1\|_\mathbb{F}\|\wbTheta_2\|_\mathbb{F}=\|\bD_1\|_\mathbb{F}\|\bD_2\|_\mathbb{F}=c\|\bD_2\|^2_\mathbb{F}$.
Therefore,
\[
\|\wbTheta_1^T\wbTheta_2\|_\star=\|\wbTheta_1\|_\mathbb{F}\|\wbTheta_2\|_\mathbb{F},
\]
and hence $\rho(\bX_1,\bX_2)=1$.

\subsection{Proof of Corollary 3.3}
Under Model (2.1) in the main paper, with the correctly specified ranks and the identifiability conditions, we have $\col(\wbTheta_1)=\col((\bU_0,\bU_1))$ and $\col(\wbTheta_1)=\col((\bU_0,\bU_2))$.
Thus,  $\rho(\bX_1,\bX_2)=0$ if and only if $\bU_0=\0$ and $\bU_1^T\bU_2=\0$.
This proves (i) of Corollary 3.3.

if $\bU_1=\0$ and $\bU_2=\0$, we have $\wbTheta_1=\bU_0\bV_1^T$ and $\wbTheta_2=\bU_0\bV_2^T$.
In particular, let $\bD_0=\bU_0^T\bU_0$.
From the identifiability conditions we know $\bD_0$ is a diagonal matrix with positive diagonal values.
We further set $\bL=\bU_0\bD_0^{-{1\over2}}$, $\bR_1={1\over\sqrt{c}}\bV_1$ and $\bM_1=\sqrt{c}\bD_0^{1\over2}$.
Under the additional condition $\bV_1^T\bV_1=c\bI$ ($0<c<1$), we know $\bL^T\bL=\bR_1^T\bR_1=\bI$ and $\bM_1$ is a diagonal matrix with positive diagonal values. Similarly, we set $\bR_2={1\over\sqrt{1-c}}\bV_2$ and $\bM_2=\sqrt{1-c}\bD_0^{1\over2}$.
Thus, $$\wbTheta_1=\bU_0\bV_1^T=\bL\bM_1\bR_1^T,\quad \wbTheta_2=\bU_0\bV_2^T=\bL\bM_2\bR_2^T$$ are the SVD of $\wbTheta_1$ and $\wbTheta_2$, respectively.
Namely, $\wbTheta_1$ and $\wbTheta_2$ have the same left singular vectors (i.e., $\bL$), and the singular values are proportional (i.e., $\bM_1=\sqrt{c\over 1-c}\bM_2$).
From the previous result, we know $\rho(\bX_1,\bX_2)=1$.
This proves (ii) of Corollary 3.3.

\subsection{Examples of Association Coefficients}

To better understand the association coefficient and the conditions under which it is equal to one, we provide a couple of examples under Model (2.1) when the identifiability conditions are satisfied.
In particular, we assume there is only joint structure in the data, i.e., $\bU_1=\0$ and $\bU_2=\0$.

First, we consider the case where $r_0=1$ and the joint score and loading are $\bu_0$ and $(\bv_1^T,\bv_2^T)^T$, respectively.
The expression of the association coefficient becomes
\[
\rho(\bX_1,\bX_2)={\|\bv_1\bu_0^T\bu_0\bv_2^T\|_\star\over\|\bu_0\bv_1^T\|_\mathbb{F}\|\bu_0\bv_2^T\|_\mathbb{F}}.
\]
The numerator is $\|\bv_1\|_\mathbb{F}\|\bv_2\|_\mathbb{F}\|\bu_0\|^2_\mathbb{F}$ which is equivalent to the denominator.
Namely, $\rho(\bX_1,\bX_2)=1$.
In other words, when the individual structure does not exist and the joint structure is unit-rank, the association coefficient is always equal to one.

Now consider the case $r_0>1$. We remark that the absence of the individual structure is no longer sufficient for $\rho(\bX_1,\bX_2)=1$.
The reason lies in the fact that although the joint loadings in $(\bV_1^T, \bV_2^T)^T$ are orthonormal, the individual matrices $\bV_1$ and $\bV_2$ are unconstrained.
If, after reordering the columns, $(\bV_1^T, \bV_2^T)^T$ presents a $2\times 2$ block-wise pattern with large values in the diagonal blocks and small (but not all zero) values in the off-diagonal blocks, the nominal joint structure essentially captures the individual patterns. Correspondingly, the singular values of $\wbTheta_1^T\wbTheta_2$ compared to the separate Frobenius norms of $\wbTheta_1$ and $\wbTheta_2$ are small, and hence the association coefficient is small.
We emphasize that this is a desired property of the newly defined association coefficient, because it automatically reduces the risk of overestimation of the strength of association when the joint and individual ranks are misspecified due to some numerical noise.

As a toy example, consider the case where there is no individual structure, $r_0 = 2$, $p_1 = p_2=2$, $n=3$ and the decomposition of  $(\wbTheta_1,\wbTheta_2)$ is
\begin{align*}
(\wbTheta_1, \wbTheta_2) = \bU_0 (\bV_1^T, \bV_2^T)=
\bpm
2 & 1\\
-2 & 1\\
0 & -2
\epm
\bpm
5/\sqrt{50.02} & 5/\sqrt{50.02} & 0.1/\sqrt{50.02} & -0.1/\sqrt{50.02}\\
0.1/\sqrt{50.02} &-0.1/\sqrt{50.02} & 5/\sqrt{50.02} & 5/\sqrt{50.02}\\
%
\epm.
\end{align*}
In this example, $\bV_1$ has much larger norm of the first column than the second column, while $\bV_2$ is the opposite. Conceptually, this indicates that $\wbTheta_1$ is primarily formed by the first column of $\bU_0$, and $\wbTheta_2$ is primarily formed by the second column of $\bU_0$.
Hence, while $\bU_0$ is deemed shared across both matrices, the weights put on different columns are quite different.
In other words, $\bU_0$ more likely captures the individual structure.
The association coefficient of the data is only 0.0404, which well reflects the fact.

In contrast, consider
\begin{align*}
(\wbTheta_1, \wbTheta_2) = \bU_0 (\bV_1^T, \bV_2^T)=
\bpm
2 & 1\\
-2 & 1\\
0 & -2
\epm
\bpm
0.1/\sqrt{1.5} & 0.2/\sqrt{1.5} & 0.8\sqrt{1.5} & 0.9/\sqrt{1.5}\\
-0.2/\sqrt{1.5} &0.1/\sqrt{1.5} & -0.9\sqrt{1.5} & 0.8/\sqrt{1.5} \epm.
\end{align*}
Although the scale of $\bV_1$ is generally smaller than that of $\bV_2$, the respective column norms are homogeneous, indicating $\bU_0$ is the truly joint structure.
The association coefficient for this example is equal to 1.

\section{GLM with Heterogeneous Link Functions}
Let $\by=(y_1,\cdots,y_n)^T\in\real^n$ denote a vector of random variables with potentially heterogenous distributions from the exponential family. In particular, assume the pdf of $y_i$ is $f_i(y_i)=h_i(y_i)\exp(y_i\theta-b_i(\theta))$, where $b_i(\cdot)$ is the corresponding cumulant function.
Let $\bX=(\bx_{(1)},\cdots,\bx_{(n)})^T$ be an $n\times p$ design matrix and $\bbeta\in\real^p$ be an unknown coefficient vector.
Suppose our goal is to fit the following GLM
\[
\Ex(y_i)=g_i^{-1}(\bx_{(i)}^T\bbeta), \ i=1,\cdots,n;
\]
where $g_i(\cdot)$ is an appropriate link function for the $i$th observation.

Following the derivation of the IRLS algorithm \citep{mccullagh1989generalized} verbatim, we obtain that each iteration solves the following weighted least square problem:
\be\label{IRLS}
\min_{\bbeta} \|\bW^{1\over2}\by^\star-\bW^{1\over2}\bX\bbeta\|_\mathbb{F}^2,
\ee
where $\bW$ is a diagonal weight matrix and $\by^\star=(y_1^\star,\cdots,y_n^\star)^T$ is an induced response vector.
More specifically,
\[
\bW=\diag\left({1\over b_1''(\hat{\theta}_1) {g_1'}^2(\hat{\mu}_1)},\cdots,{1\over b_n''(\hat{\theta}_n) {g_n'}^2(\hat{\mu}_n)}\right),
\]
and
\[
y_i^\star=\bx_{(i)}^T\hat{\bbeta}+(y_i-\hat{\mu}_i)g'_i(\hat{\mu}_i),\quad i=1,\cdots,n,
\]
where $\hat{\bbeta}$ is the coefficient estimate from the previous iteration, $\hat{\mu}_i=g_i^{-1}(\bx_{(i)}^T\hat{\bbeta})$, and $\hat{\theta}_i=b_i'^{-1}(\hat{\mu}_i)$.
Thus, by iteratively solving \eqref{IRLS}, we obtain the maximum likelihood estimate of $\bbeta$.

\section{Details of the One-Step Approximation Algorithm}\label{one-step}

To further alleviate the computational burden of the double-iterative model fitting algorithm, we substitute the IRLS algorithm for the GLM model fitting with a one-step approximation with warm start.
More specifically, to estimate each parameter, we use the estimate from the previous iteration as the initial value to calculate the induced response and weights as in the standard IRLS algorithm, and solve a weighted least square problem exactly once.
The obtained estimate, after proper normalization, is used in the next iteration.
As a result, there is only one layer of iteration in the entire algorithm.

More specifically, in each iteration, we update the model parameter estimates sequentially, following the order:
\[
\bU_1\rightarrow\{\bmu_1,\bA_1\}\rightarrow\bU_2\rightarrow\{\bmu_2,\bA_2\}\rightarrow\{\bmu_1,\bV_1\}\rightarrow\{\bmu_2,\bV_2\}\rightarrow\bU_0.
\]
We remark that any change of the order does not affect the convergence of the algorithm.
In addition, whether to update the estimate of the intercepts ($\bmu_1$ and $\bmu_2$) twice as is, or just once with the individual loadings, or just once with the joint loadings, has little effect on the final results. Thus, we focus on the above order hereafter.

We denote the estimates from the previous iteration by $\{\widetilde{\bmu}_1,\widetilde{\bmu}_2,\widetilde{\bU}_0,\widetilde{\bU}_1,\widetilde{\bU}_2,\widetilde{\bV}_1,\widetilde{\bV}_2,\widetilde{\bA}_1,\widetilde{\bA}_2\}$. To estimate each row of $\bU_1$ (i.e., $\bu_{1,(i)}$), in the original algorithm we propose to fit the following GLM
\[
\Ex(\bx_{1,(i)})=b_1'(\btheta_{1,(i)}), \mbox{ and } \btheta_{1,(i)}=\widetilde{\bmu}_1+\widetilde{\bV}_1\widetilde{\bu}_{0,(i)}+\widetilde{\bA}_1\bu_{1,(i)},
\]
where $b_1'(\cdot)$ represent an entrywise function.
The one-step approximation algorithm, which we shall elaborate here, alleviates computation by performing just one step of the IRLS algorithm.
More specifically, let $\widetilde{\btheta}_{1,(i)}=\widetilde{\bmu}_1+\widetilde{\bV}_1\widetilde{\bu}_{0,(i)}+\widetilde{\bA}_1\widetilde{\bu}_{1,(i)}$.
We only need to solve the following weighted least square problem
\be\label{U1}
\min_{\bu_{1,(i)}} \|\bW^{1\over2}\by^\star-\bW^{1\over2}\widetilde{\bA}_1\bu_{1,(i)}\|_\mathbb{F}^2,
\ee
where
\[
\bW=\diag\left( b_1''(\widetilde{\btheta}_{1,(i)})\right), \mbox{ and }
\by^\star=\widetilde{\bA}_1\widetilde{\bu}_{1,(i)}+\left\{\bx_{1,(i)}-b_1'(\widetilde{\btheta}_{1,(i)})\right\}\cdot {1\over b_1''(\widetilde{\btheta}_{1,(i)})}.
\]
Similar to the original algorithm, the estimation of different rows of $\bU_1$ can be easily parallelized.
Once every row is estimated, we update $\widetilde{\bU_1}$ to be the latest estimates.

To estimate $\{\bmu_1,\bA_1\}$, let us denote $\widetilde{\btheta}_{1,j}=\widetilde{\mu}_{1j}\1+\widetilde{\bU}_0\widetilde{\bv}_{1,(j)}+\widetilde{\bU}_1\widetilde{\ba}_{1,(j)}$, and  solve the following weighted least square problem
\be\label{A1}
\min_{\mu_{1j},\ba_{1,(j)}} \|\bW^{1\over2}\by^\star-\bW^{1\over2}(\mu_{1j}\1+\widetilde{\bU}_1\ba_{1,(j)})\|_\mathbb{F}^2,
\ee
where
\[
\bW=\diag\left( b_1''(\widetilde{\btheta}_{1,j})\right), \mbox{ and }
\by^\star=(\widetilde{\mu}_{1j}\1+\widetilde{\bU}_1\widetilde{\ba}_{1,(j)})+\left\{\bx_{1,j}-b_1'(\widetilde{\btheta}_{1,j})\right\}\cdot {1\over b_1''(\widetilde{\btheta}_{1,j})}.
\]
Again, once estimated, we update $\widetilde{\bmu}_1$ and $\widetilde{\bA}_1$ to be the latest estimates.
Almost identically, we can update the estimates of $\bU_2$, $\bmu_2$, and $\bA_2$.

To estimate $\{\bmu_1,\bV_1\}$, we exploit the same expression of $\widetilde{\btheta}_{1,j}$, and solve the following weighted least square problem
\be\label{V1}
\min_{\mu_{1j},\bv_{1,(j)}} \|\bW^{1\over2}\by^\star-\bW^{1\over2}(\mu_{1j}\1+\widetilde{\bU}_0\bv_{1,(j)})\|_\mathbb{F}^2,
\ee
where
\[
\bW=\diag\left( b_1''(\widetilde{\btheta}_{1,j})\right), \mbox{ and }
\by^\star=(\widetilde{\mu}_{1j}\1+\widetilde{\bU}_0\widetilde{\bv}_{1,(j)})+\left\{\bx_{1,j}-b_1'(\widetilde{\btheta}_{1,j})\right\}\cdot {1\over b_1''(\widetilde{\btheta}_{1,j})}.
\]
Similarly, we estimate $\bmu_2$ and $\bV_2$.

Finally, we estimate $\bU_0$. Let us denote $\widetilde{\btheta}_{0,(i)}=\left(\widetilde{\bmu}_1^T+\widetilde{\bu}^T_{1,(i)}\widetilde{\bA}_1^T\ , \ \widetilde{\bmu}_2^T+\widetilde{\bu}^T_{2,(i)}\widetilde{\bA}_2^T\right)^T +\widetilde{\bV}_0\widetilde{\bu}_{0,(i)}$.
Furthermore, with a slight abuse of notation,  we use $b_0(\cdot)$ to denote an entrywise function mapping $\real^{p_1+p_2}$ to $\real^{p_1+p_2}$, with the first $p_1$ functions being $b_1:\real\mapsto\real$, and the last $p_2$ functions being $b_2:\real\mapsto\real$.
Correspondingly, $b_0'(\cdot)$ and $b_0''(\cdot)$ denote the entrywise first and second order derivative functions of $b_0(\cdot)$, respectively.
Subsequently, we solve the following weighted least square problem
\be\label{U0}
\min_{\bu_{0,(i)}} \|\bW^{1\over2}\by^\star-\bW^{1\over2}\widetilde{\bV}_0\bu_{0,(i)}\|_\mathbb{F}^2,
\ee
where
\[
\bW=\diag\left( b_0''(\widetilde{\btheta}_{0,(i)})\right), \mbox{ and }
\by^\star=\widetilde{\bV}_0\widetilde{\bu}_{0,(i)}+\left\{(\bx_{1,(i)}^T,\bx_{2,(i)}^T)^T-b_0'(\widetilde{\btheta}_{0,(i)})\right\}\cdot {1\over b_0''(\widetilde{\btheta}_{0,(i)})}.
\]

At the end of each iteration, we normalize the estimated parameters following the same procedure as in the main paper. Consequently, the obtained parameters satisfy the identifiability conditions.
After each iteration, we calculate the difference of the log likelihood values between the current estimates and the previous estimates. We stop the iterations when the difference becomes sufficiently small. Although there is no proof that the one-step approximation algorithm will increase the likelihood value in each iteration as the original algorithm does, we observe that it typically converges quickly. A more rigorous proof of convergence needs further investigation.
The pseudo code of the one-step approximation algorithm is presented in Algorithm \ref{alg2}.

\begin{algorithm}[h]
\caption{The One-Step Approximation Algorithm for Model Fitting}\label{alg2}
\begin{algorithmic}
\State Initialize $\{\bmu_1,\bmu_2,\bU_0,\bU_1,\bU_2,\bV_1,\bV_2,\bA_1,\bA_2\}$;
\While {The log likelihood difference has not reached convergence}
\bi
\item Estimate $\bu_{1,(i)}$ by solving \eqref{U1} for $i=1,\cdots,n$ in parallel;
\item Estimate $\{\mu_{1j},\ba_{1,(j)}\}$ by solving \eqref{A1} for $j=1,\cdots,p_1$ in parallel;
\item Estimate $\bu_{2,(i)}$ the same way as one estimates $\bu_{1,(i)}$;
\item Estimate $\{\mu_{2j},\ba_{2,(j)}\}$ the same way as one estimates $\{\mu_{1j},\ba_{1,(j)}\}$;
\item Estimate $\{\mu_{1j},\bv_{1,(j)}\}$ by solving \eqref{V1} for $j=1,\cdots,p_1$ in parallel;
\item Estimate $\{\mu_{2j},\bv_{2,(j)}\}$ the same way as one estimates $\{\mu_{1j},\bv_{1,(j)}\}$;
\item Estimate $\bu_{0,(i)}$ by solving \eqref{U0} for $i=1,\cdots,n$ in parallel;
\item Normalize the estimated parameters to retrieve the identifiability conditions;
\item Calculate the log likelihood value of the new parameter estimates.
\ei
\EndWhile
\end{algorithmic}
\end{algorithm}

\section{Rank Estimation}

There has been a large body of literature on selecting ranks for matrix factorization problems and determining the number of components in factor models under the Gaussian assumption \citep{bai2002determining,kritchman2008determining,owen2009bi}.
However, none of the methods directly extends to non-Gaussian data.
Moreover, little has been studied for the rank estimation of more than one data set.

In Section \ref{sec:CV}, we develop an $N$-fold cross validation (CV) approach to estimate the rank of the column-centered natural parameter matrix underlying a non-Gaussian data set. The approach flexibly accommodates a data matrix from a single distribution, or a data matrix consisting of mixed variables from multiple distributions.
In Section \ref{sec:two-step}, we devise a two-step procedure to estimate the joint and individual ranks $(r_0,r_1,r_2)$ in Model (2.1) in the main paper.
In Section \ref{sec:CVnum}, we validate the two-step procedure using different simulation examples described in Section 6.1 of the main paper.
Finally, in Section \ref{sec:real}, we apply the two-step procedure to estimate the model ranks for the CAL500 data.

\subsection{$N$-Fold CV}\label{sec:CV}
Let $\bX$ represent an $n\times p$ data matrix, where the entries are independently distributed and may follow heterogeneous distributions from the exponential family.
Let $\bTheta=\1\bmu^T+\wbTheta$ represent the underlying natural parameter matrix with $\wbTheta$ being the column-centered structure.
The goal is to estimate the rank of $\wbTheta$.

The idea  stems from the CV procedure for estimating the number of principal components in factor models \citep{wold1978cross,bro2008cross,josse2012selecting}. Here we generalize it to the exponential family, and furthermore, to mixed data types.
The general procedure is as follows. First, we randomly split the entries of $\bX$ into $N$ blocks of  roughly equal size. Each time, we use $N-1$ blocks of data to estimate the natural parameter matrices with different candidate ranks. With each estimated natural parameter matrix, we predict the left-out entries with the corresponding expectations, and calculate the sum of squared Pearson residuals of those entries. The CV score is the sum of squares divided by the number of entries in this block.
We repeat this procedure for all $N$ blocks, and take the average or median of the $N$ CV scores as the overall score for each candidate rank. The rank with the minimum overall score is selected.

More specifically, let $x_{ij}$ and $\theta_{ij}$ be the $ij$th entries of $\bX$ and $\bTheta$, respectively.
The pdf of $x_{ij}$ is
\[
f_{ij}(x_{ij}|\theta_{ij})=h_{ij}(x_{ij})\exp\{x_{ij}\theta_{ij}-b_{ij}(\theta_{ij})\},\quad i=1,\cdots,n; j=1,\cdots,p,
\]
where $f_{ij}(\cdot)$ is the pdf for $x_{ij}$ with potentially heterogeneous normalization function $h_{ij}(\cdot)$ and cumulant function $b_{ij}(\cdot)$.
we first randomly split the entries of $\bX$.
Let $\bx^{[l]}$ denote the vector of left-out entries in the $l$th block ($l=1,\cdots,N$), and $\bX^{[-l]}$ denote the remaining data matrix where the values of the left-out entries are missing.
In particular, we require that none of the rows or columns in $\bX^{[-l]}$ is entirely missing.
Otherwise, we manually modify the partition or simply re-split the data.
The requirement is easily satisfied in practice as long as $N$ is moderately large (e.g., $N\geq 5$).

Next, we use $\bX^{[-l]}$ to estimate a natural parameter matrix with rank $r$ for the column-centered structure.
Let $\bTheta=\1\bmu^T+\wbTheta$ denote the natural parameter matrix, where ${\wbTheta}=\bU\bV^T$ is a rank-$r$ matrix with $\1^T\bU=\0$ and $\bU\in\real^{n\times r},\bV\in\real^{p\times r}$.
We exploit an alternating procedure, similar to the model fitting algorithm, to estimate the parameters $\{\bmu,\bU,\bV\}$ via parallel GLMs.
Moreover, the one-step approximation idea described in Section \ref{one-step} is readily applicable to facilitate computation.
When $\bU$ is fixed, we fit a model to the 
observed values in each column of $\bX^{[-l]}$ 
to estimate each entry of $\bmu$ (i.e., $\mu_j$) and each row of $\bV$ (i.e., $\bv_{(j)}$). Specifically, denote $\widetilde{\theta}_{ij}=\widetilde{\mu}_j+\widetilde{\bu}_{(i)}^T\widetilde{\bv}_{(j)}$, where the parameters with the tilde symbol are estimated from the previous iteration. To estimate $\mu_j$ and $\bv_{(j)}$, we shall solve
\bes
\min_{\mu_{j},\bv_{(j)}} \|\bW^{1\over2}\by^\star-\bW^{1\over2}(\mu_{j}\1+\widetilde{\bU}\bv_{(j)})\|_\mathbb{F}^2,
\ees
where $\bW$ is an $n\times n$ diagonal matrix with the $i$th diagonal value being
\bes
w_{ii}=\left\{
\begin{aligned}
&b_{ij}''(\widetilde{\theta}_{ij}),\quad \mbox{if $x_{ij}$ is observed,}\\
&0,\quad\quad\ \ \quad \mbox{otherwise},
\end{aligned}\right.
\ees
and $\by^\star$ is a length-$n$ vector with the $i$th value being $y_{i}^\star=\widetilde{\theta}_{ij}+\left\{x_{ij}-b_{ij}'(\widetilde{\theta}_{ij})\right\}/b_{ij}''(\widetilde{\theta}_{ij})$.
Similarly, when $\{\bmu,\bV\}$ is fixed, we fit a model to the observed values in each row of $\bX^{[-l]}$ to estimate each row of $\bU$ (i.e., $\bu_{(i)}$).
With the same notation of $\widetilde{\theta}_{ij}$, we shall solve
\bes
\min_{\bu_{(i)}} \|\bW^{1\over2}\by^\star-\bW^{1\over2}\widetilde{\bV}\bu_{(i)}\|_\mathbb{F}^2,
\ees
where $\bW$ is a $p\times p$ matrix with the $j$th diagonal value being
\bes
w_{jj}=\left\{
\begin{aligned}
&b_{ij}''(\widetilde{\theta}_{ij}),\quad \mbox{if $x_{ij}$ is observed,}\\
&0,\quad\quad\ \ \quad \mbox{otherwise},
\end{aligned}\right.
\ees
and $\by^\star$ is a length-$p$ vector with the $j$th value being
$y_{j}^\star=\left(\widetilde{\theta}_{ij}-\widetilde{\mu}_j\right)+\left\{(x_{ij}-b_{ij}'(\widetilde{\theta}_{ij})\right\}/ b_{ij}''(\widetilde{\theta}_{ij})$.
We alternate between the two steps until convergence.
Consequently, we obtain the estimate of a natural parameter matrix with rank-$r$ column-centered structure.

Let $\widehat{\bTheta}^{[-l]}_r$ represent the estimated natural parameter matrix from $\bX^{[-l]}$ with rank $r$ for the column-centered structure.
The Pearson residual for $x_{ij}$ is defined as
\[
R_{ij}={x_{ij}-b_{ij}'(\widehat{\theta}^{[-l]}_{r,ij})\over \sqrt{b''_{ij}(\widehat{\theta}^{[-l]}_{r,ij})}},
\]
where $\widehat{\theta}^{[-l]}_{r,ij}$ is the $ij$th entry of $\widehat{\bTheta}^{[-l]}_{r}$.
The CV score for rank $r$ in the $l$th fold is calculated as the summation of the squared Pearson residuals for the entries in $\bx^{[l]}$, divided by the number of entries in $\bx^{[l]}$.
Similarly, we can calculate the CV scores for different ranks and in different folds.
Finally, we compare the average or the median of the CV scores across different folds for different candidate ranks, and select the rank with the minimum score.

\subsection{Two-Step Rank Estimation Procedure}\label{sec:two-step}
To estimate the joint and individual ranks in Model (2.1) of the paper main, we devise a two-step procedure.
First, we apply the CV procedure described in Section \ref{sec:CV} to $\bX_1$, $\bX_2$, and the concatenated data set $(\bX_1,\bX_2)$, respectively.
We obtain the estimates of the ranks of the column-centered natural parameter matrices $\bTheta_1$, $\bTheta_2$, and $(\bTheta_1,\bTheta_2)$ as $r_1^\star$, $r_2^\star$, and $r_0^\star$.
According to the identifiability conditions in Section 2.2 of the main paper, we know that $r_0^\star=r_0+r_1+r_2$, $r_1^\star=r_0+r_1$, and $r_2^\star=r_0+r_2$.
Therefore, in the second step, by solving the linear equations, we obtain the estimate of the joint and individual ranks $(r_0,r_1,r_2)$ as
\be\label{rank}
\widehat{r_0}=r_1^\star+r_2^\star-r_0^\star, \quad \widehat{r_1}=r_0^\star-r_1^\star, \quad \widehat{r_2}=r_0^\star-r_1^\star.
\ee
A similar procedure has been used in \cite{hellton2016integrative}.
As a result, we obtain the rank estimates for Model (2.1) in the main paper.

{\color{black} In practice, low ranks are typically preferred for the computational efficiency and interpretability.
Thus, we can set a small upper bound (i.e., 10) for $r_1^\star$ and $r_2^\star$.}
Moreover, notice that $ \max(r_1^\star,r_2^\star)\leq r_0^\star \leq r_1^\star+r_2^\star$.
One could first select $r_1^\star$ and $r_2^\star$ using the CV procedure, and then use $\max(r_1^\star, r_2^\star)$ and $r_1^\star+r_2^\star$ as the lower and upper bounds for the CV candidate set of $r_0^\star$.



\subsection{Numerical Studies}\label{sec:CVnum}
In this section, we validate the two-step rank estimation procedure using the four simulation settings described in Section 6.1 of the main paper.

Given two data sets $\bX_1$ and $\bX_2$ in each simulation setting,
we first estimate the ranks of the underlying column-centered natural parameter matrices of $\bX_1$, $\bX_2$, and the concatenated data $(\bX_1,\bX_2)$, respectively.
According to the setup, the true ranks are 4, 4, and 6.
We let the candidate set of the ranks for the individual data be $\{1,2,3,4,5,6\}$, and use the selected individual ranks to determine the range of the candidate set for the concatenated data.
We apply the 10-fold CV method in each case, and the results are presented in Figures \ref{rank1}--\ref{rank4}, each corresponding to a single simulation run in each setting.

\begin{figure}[htbp]
  \centering
  \includegraphics[width=6in,height=3in]{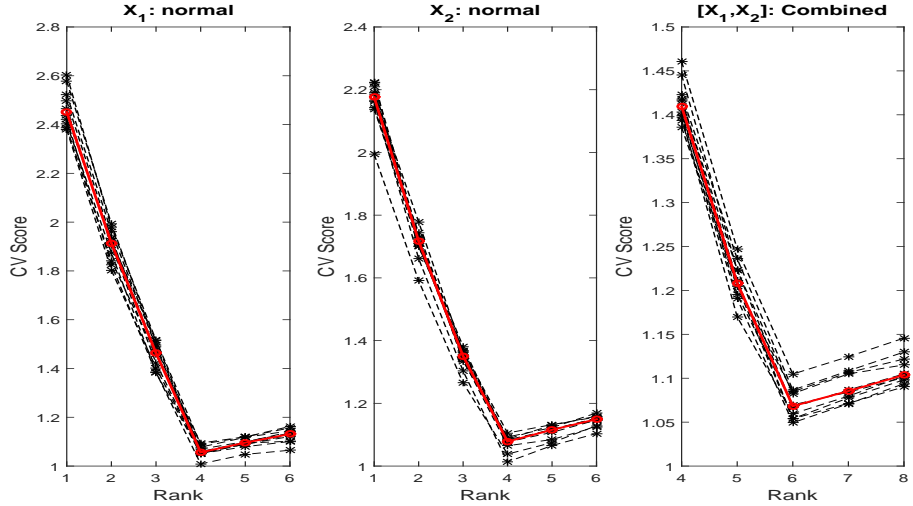}
  \caption{Rank selection under Setting 1 (Gaussian-Gaussian). From left to right is the 10-fold CV score plot for $\bX_1$, $\bX_2$, and $(\bX_1,\bX_2)$ respectively. In each plot, a dashed line with asterisks corresponds to one fold of CV; the solid line with circles correspond to the median CV scores.}\label{rank1}
\end{figure}

\begin{figure}[htbp]
  \centering
  \includegraphics[width=6in,height=3in]{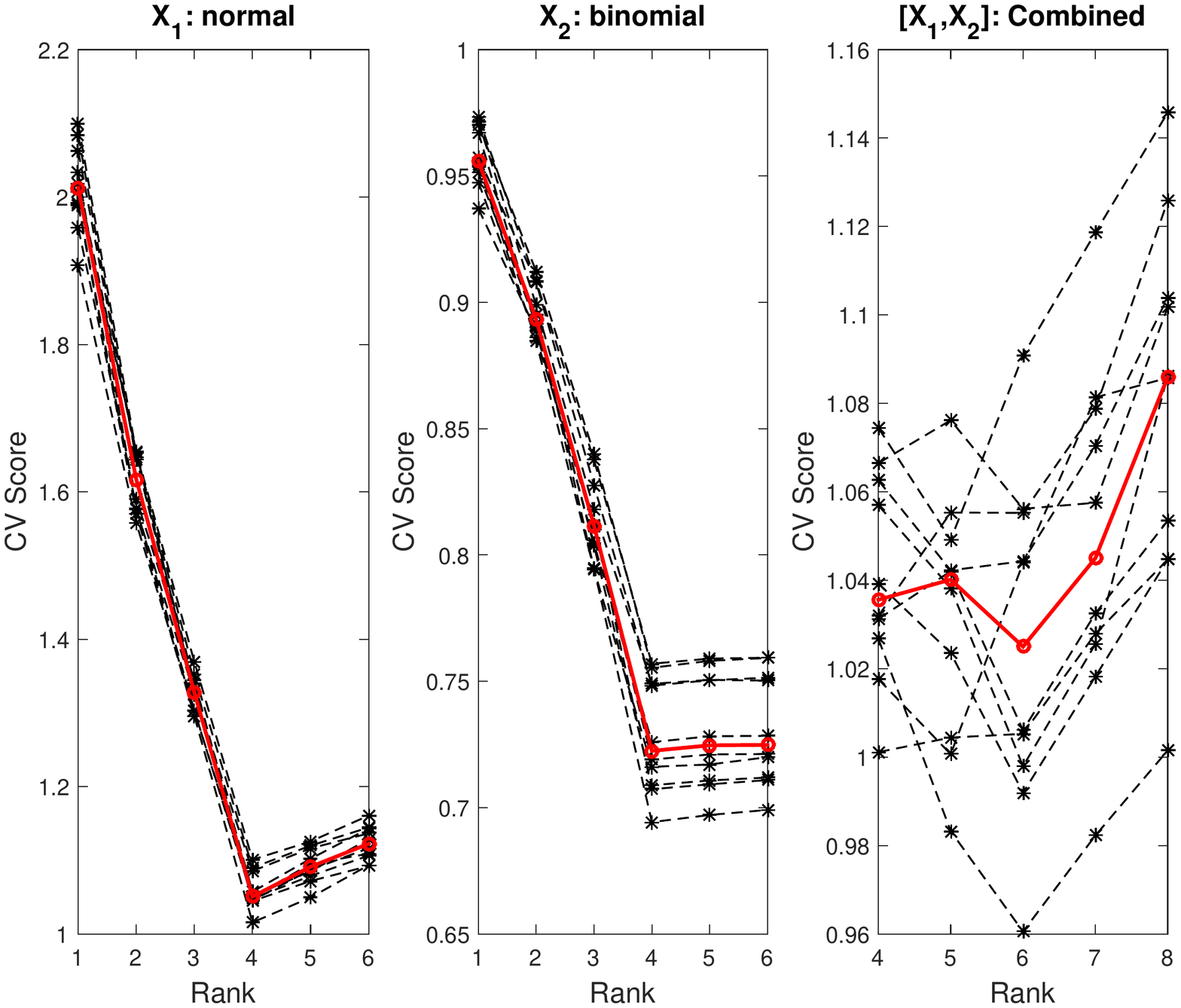}
  \caption{Rank selection under Setting 2 (Gaussian-Bernoulli). From left to right is the 10-fold CV score plot for $\bX_1$, $\bX_2$, and $(\bX_1,\bX_2)$ respectively. In each plot, a dashed line with asterisks corresponds to one fold of CV; the solid line with circles correspond to the median CV scores.}\label{rank2}
\end{figure}

\begin{figure}[htbp]
  \centering
  \includegraphics[width=6in,height=3in]{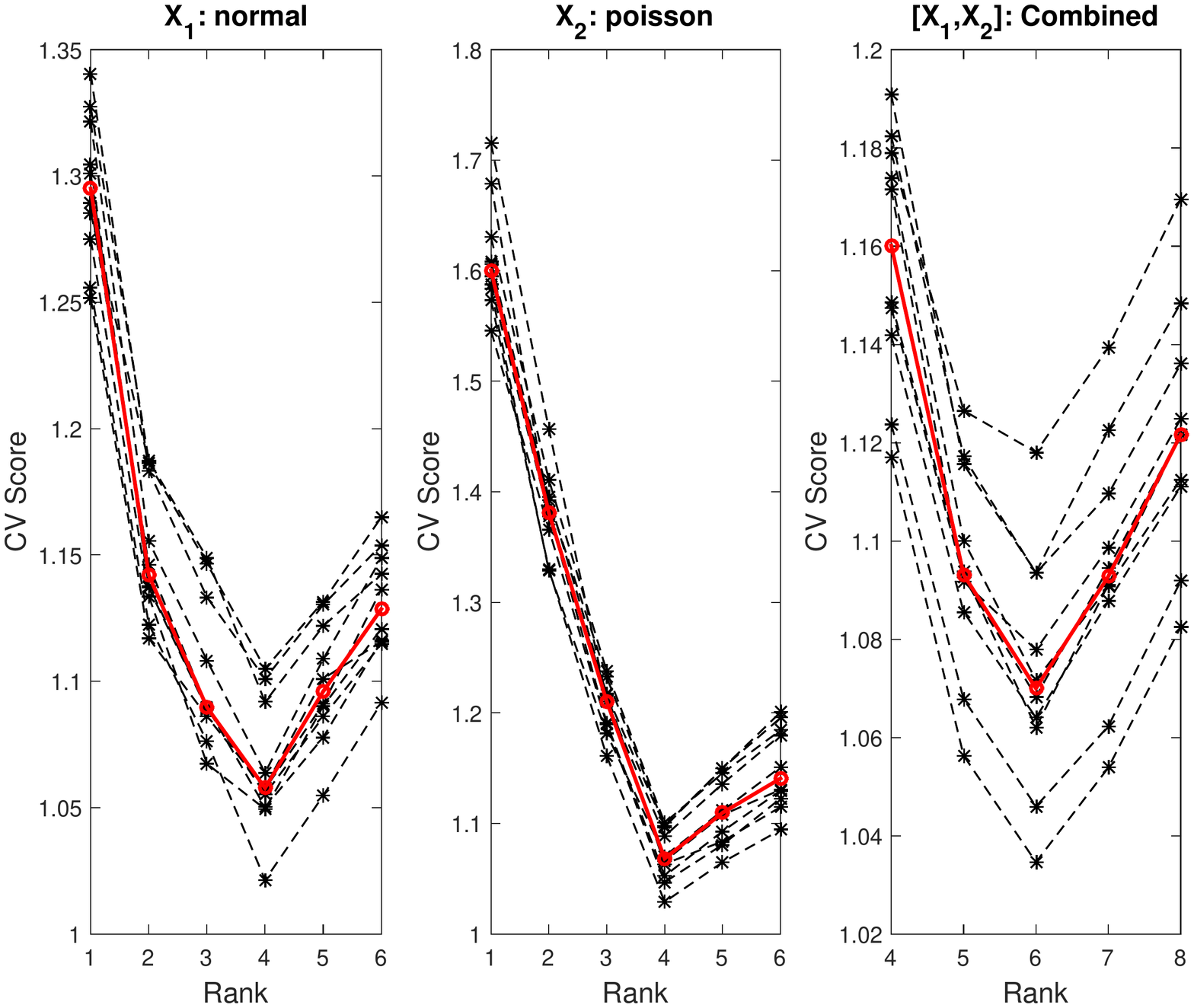}
  \caption{Rank selection under Setting 3 (Gaussian-Poisson). From left to right is the 10-fold CV score plot for $\bX_1$, $\bX_2$, and $(\bX_1,\bX_2)$ respectively. In each plot, a dashed line with asterisks corresponds to one fold of CV; the solid line with circles correspond to the median CV scores.}\label{rank3}
\end{figure}

\begin{figure}[htbp]
  \centering
  \includegraphics[width=6in,height=3in]{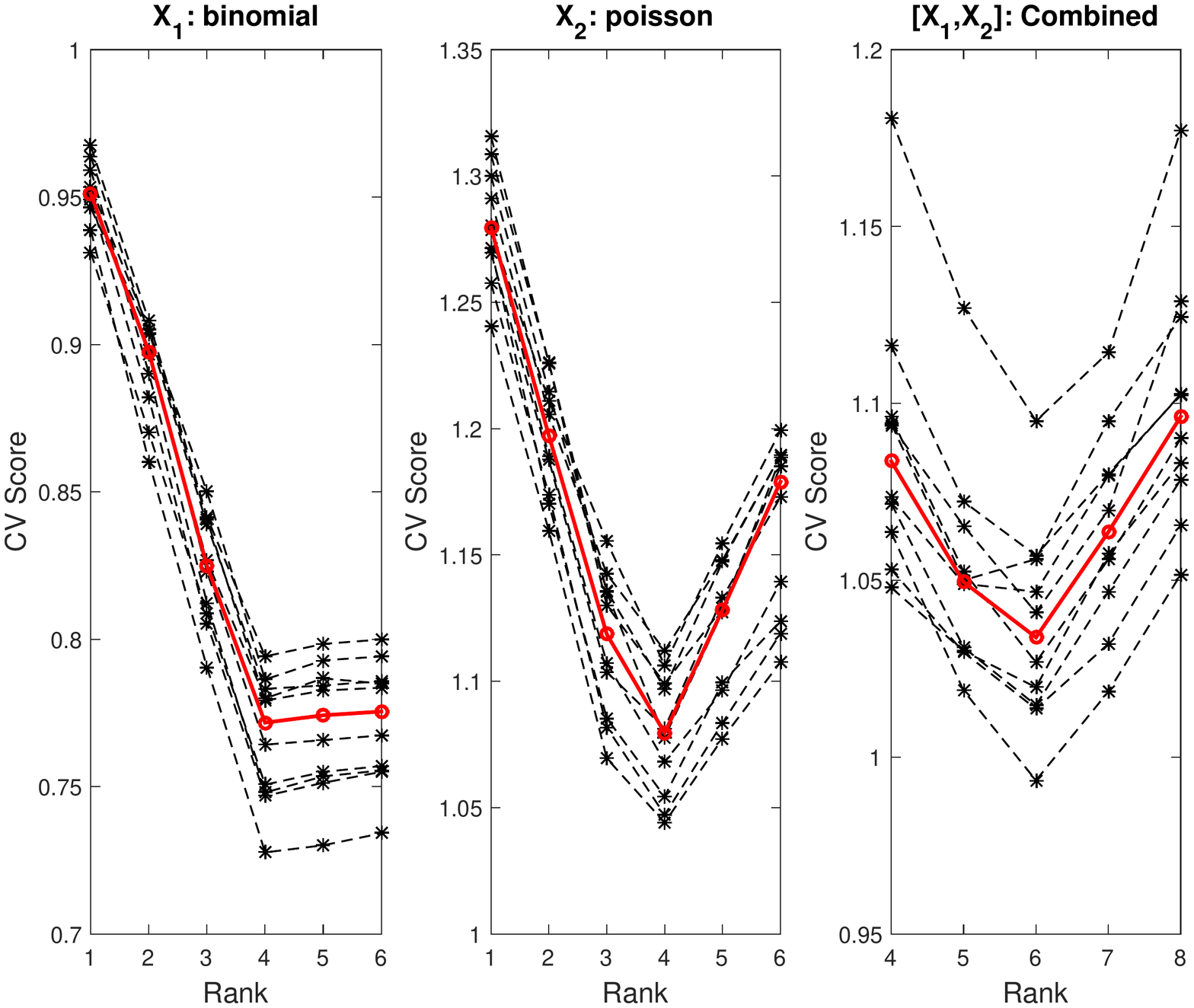}
  \caption{Rank selection under Setting 4 (Bernoulli-Poisson). From left to right is the 10-fold CV score plot for $\bX_1$, $\bX_2$, and $(\bX_1,\bX_2)$ respectively. In each plot, a dashed line with asterisks corresponds to one fold of CV; the solid line with circles correspond to the median CV scores.}\label{rank4}
\end{figure}

Overall, the 10-fold CV procedure works very well for various data types in different settings.
Cross validation for each block of data almost always correctly identifies the true ranks, except for a couple of times for mixed-type data involving Bernoulli data in Setting 2 and Setting 4.
We also notice that for purely Bernoulli data (e.g., the middle panel in Figure \ref{rank2} and the left panel in Figure \ref{rank4}), the CV scores tend to drop quickly before the candidate rank reaches the true rank, and stay flat afterwards.
This pattern makes it difficult to select the correct ranks for Bernoulli data.
We emphasize that in general the rank estimation for Bernoulli data is extremely difficult, because dichotomized data contain relatively scarce information about the rank of the underlying structure.
Unless the signal level (i.e., the magnitude of the natural parameters) is relatively high, it is very tricky to correctly estimate the rank for a Bernoulli data matrix.
To our best knowledge, the proposed CV method is among the first attempts to address this problem.
Given the prevalence of binary data in practice (e.g., genetic mutations, music annotations), the corresponding rank estimation problem remains an open question.

Once the separate ranks are estimated, the second step is to calculate the joint and individual model ranks using \eqref{rank}.
As a result, we obtain a unique set of joint and individual ranks for the model.
In the above simulation studies, since the selected values of the separate ranks are equal to the true values, the subsequently calculated model ranks are also consistent with the truth.

\subsection{Rank Estimation for CAL500}\label{sec:real}
We apply the  two-step procedure to estimate the model ranks for the CAL500 data.
The 10-fold CV score plots for separate data matrices and the concatenated data matrix are shown in Figure \ref{fig:rank}.
For the individual data matrices, the CV scores flatten out from rank 6 (for acoustic features) and rank 5 (for semantic annotations), respectively.
This phenomenon is probably due to the high level of noise in the data, as we observe in the simulation study in Section \ref{sec:CVnum}.
Nevertheless, we choose $r_1^\star=6$ and $r_2^\star=5$.
Subsequently, we set the range of the rank $r_0^\star$ to be 6 (i.e., $\max(r_1^\star,r_2^\star)$) to 11 (i.e., $r_1^\star+r_2^\star$) for the concatenated data.
The CV scores reach the minimum at rank 8, and hence we choose $r_0^\star=8$.
From the set of equations in \eqref{rank}, we obtain the estimated model ranks $r_0=3$, $r_1=3$, and $r_2=2$.

\begin{figure}[htbp]
  \centering
  \includegraphics[width=5.5in]{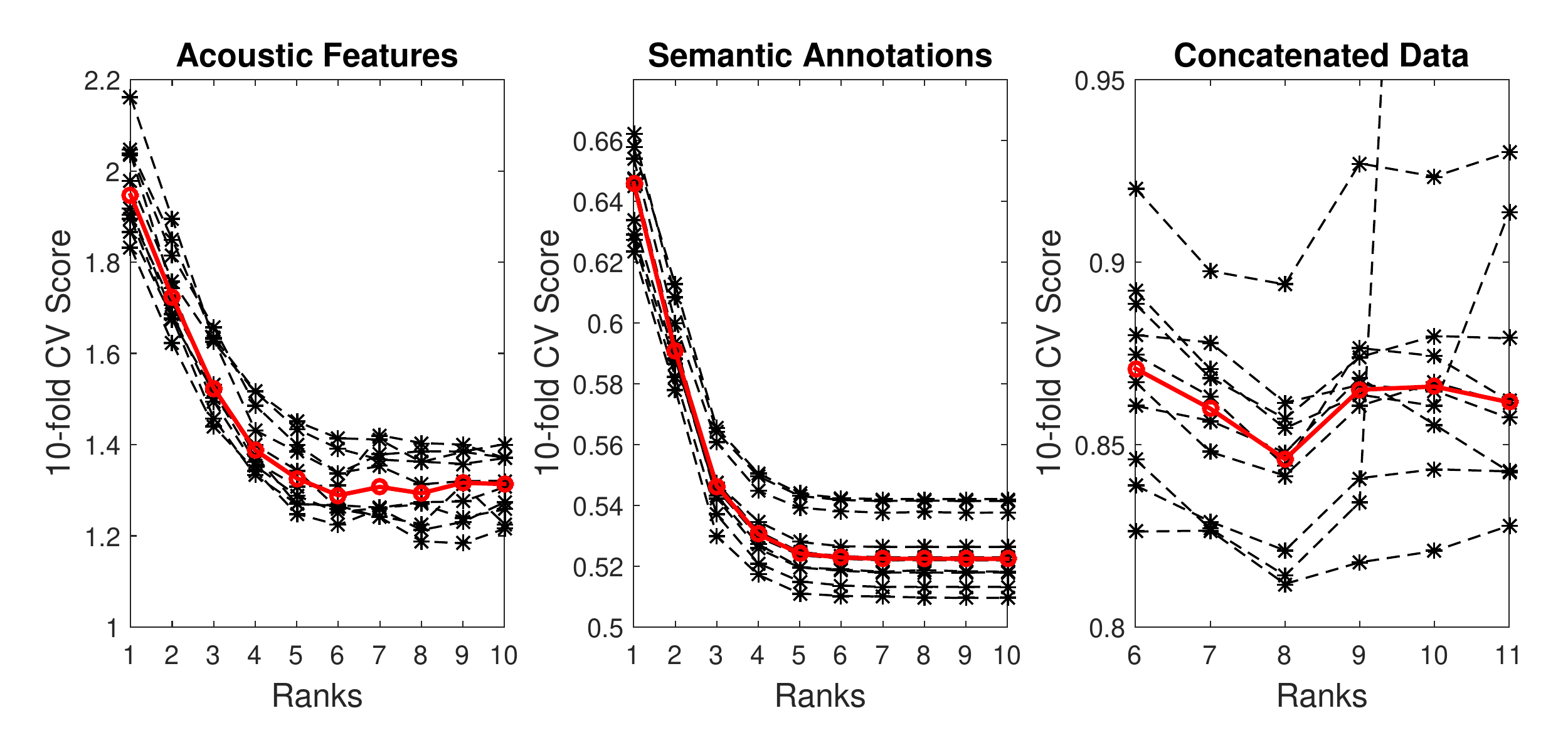}
  \caption{Rank selection for the CAL500 data. From left to right is the 10-fold CV score plot for $\bX_1$, $\bX_2$, and $(\bX_1,\bX_2)$ respectively. In each plot, a dashed line with asterisks corresponds to one fold of CV; the solid line with circles correspond to the median CV scores. }\label{fig:rank}
\end{figure}

\section{Ridge Remedy for Non-convergence for Bernoulli Data}
Sometimes the likelihood of GLM for Bernoulli random variables does not have a finite optimizer.
Consider, for example, a binary response vector $\by$ and a univariate predictor $\bx$, where $\by=\In(\bx>0)$ with $\In(\cdot)$ being an entrywise indicator function.
Let $\beta$ be the coefficient for the GLM
\[
g\{\Ex(\by)\}=\bx\beta,
\]
where $g(\cdot)$ is an entrywise link function (e.g., a logistic function).
It is easy to see that a larger value of $\beta$ generates a larger likelihood value for the GLM.
Consequently, the MLE of $\beta$ is positive infinity.

This phenomenon may lead to degenerate estimates in presence of Bernoulli data. 
It is especially non-negligible in alternating procedures, such as EPCA \citep{collins2001generalization}, and the original algorithm for GAS in the main paper. 
This is because the singularity may build up over iterations, even though the initial estimates may not be degenerate.
Without special treatment, the EPCA algorithm and the original GAS algorithm almost always fail to converge to finite values for Bernoulli data.
We emphasize that the one-step approximation algorithm effectively alleviates the problem, because in each iteration it does not implement the complete IRLS algorithm, and hence less likely to build up the singularity.
Overall, the one-step procedure is more robust against the divergence issue, but not completely immune of it.
Here we provide a universal remedy for the divergence issue for the Bernoulli data.

The idea stems from the ridge regression.
We propose to add a small ridge penalty to the Bernoulli likelihood to shrink the MLE towards zero.
As a result, the infinity is not a local optimum of the penalized likelihood any more, and the optimization algorithm will converge to a finite value.
More specifically, let $\by$ be an $n\times 1$ binary response vector and $\bX$ be an $n\times p$ design matrix. With the canonical logit link function, we propose to maximize the following penalized log likelihood function
\[
\by^T\bX\bbeta-\log\{1+\exp(\bX\bbeta)\} - {n\over 2}\lambda\|\bbeta\|^2_\mathbb{F},
\]
where $\lambda\geq0$ is a tuning parameter.
The optimization is easily implemented by a slight modification of the IRLS algorithm.
In particular, we substitute the weighted least square with the penalized weighted least square, which also bears a closed form solution.
As a result, it addresses the degeneracy issue efficiently.
Since the inclusion of the penalty will shrink the estimate towards zero, in practice, we recommend using a small tuning parameter, e.g., $\lambda=10^{-2}$ or $10^{-3}$.
Selection of the best ridge tuning parameter is beyond the scope of the paper, and remains an open question.

\section{Simulation under the Sparse Settings}
We modify the simulation settings in the main paper to obtain the corresponding sparse  settings.
In particular, we truncate the joint loadings $\bV_0=(\bV_1^T,\bV_2^T)^T$ by the 40\% quantile of the absolute values in each setting, and re-normalize them to have orthonormal columns. Consequently, we obtain a sufficiently sparse true joint loading matrix. All the other parameters are kept unchanged. Similar to the main paper, we conduct 100 simulation runs under each setting, and compare the GAS, sGAS, and EPCA-JIVE methods using various  criteria described in the paper.  The results are summarized in Figures \ref{fig:1a}--\ref{fig:4b}.

From the results we observe that the sGAS method, with variable selection in the joint loadings, outperforms the GAS method in terms of the joint loading estimation and the joint and overall structure recovery in all settings. The two methods have similar performance on the individual loading and structure estimation. This is mainly because we only introduce sparsity to the joint loadings. Hence the major advantage of the sparse method is in the joint structure estimation.
Both methods significantly outperform the EPCA-JIVE method in Settings 2--4.
When the data follow the Gaussian distribution (Setting 1), as shown in the main paper, the GAS method and the EPCA-JIVE method are essentially the same, and thus have similar performance.

\begin{figure}[h]
  \centering
  \includegraphics[width=6in]{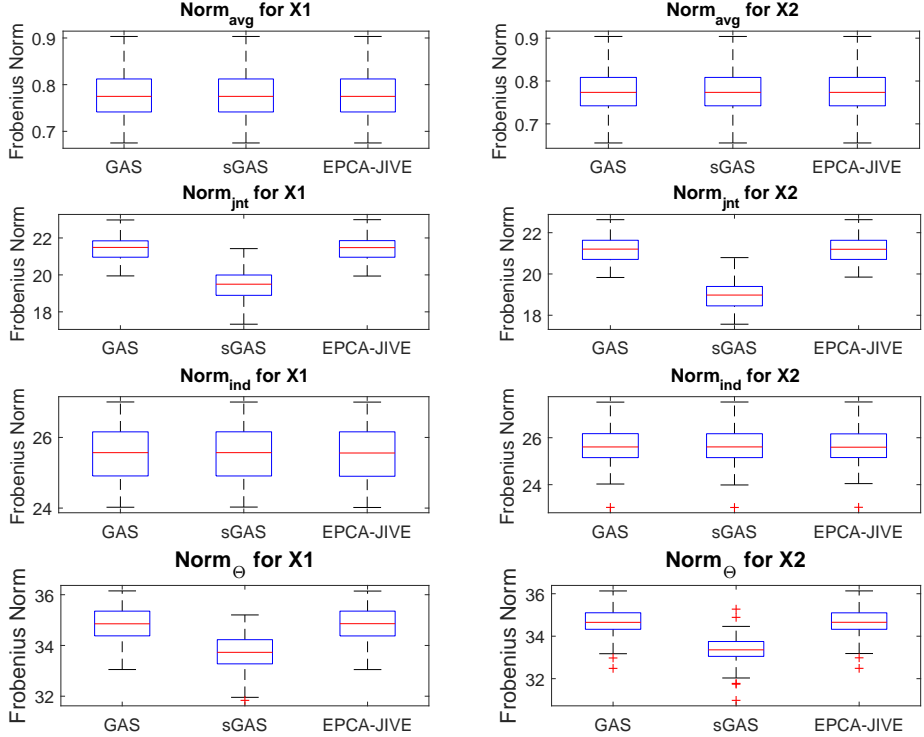}
  \caption{Sparse Setting 1 (Gaussian-Gaussian): comparison of the low-rank structure estimation accuracy among the GAS, sGAS, and EPCA-JIVE methods. The left panels are for $\bX_1$ and the right panels are for $\bX_2$. From top to bottom, we evaluate $Norm_{avg}=\|\bmu_k-\widehat{\bmu_k}\|_\mathbb{F},\ Norm_{jnt}=\|\bU_0\bV_k^T-\widehat{\bU_0}\widehat{\bV_k}^T\|_\mathbb{F},\
Norm_{ind}=\|\bU_k\bA_k^T-\widehat{\bU_k}\widehat{\bA_k}^T\|_\mathbb{F},\ Norm_{\bTheta}=\|\bTheta_k-\widehat{\bTheta_k}\|_\mathbb{F}$, respectively. }\label{fig:1a}
\end{figure}

\begin{figure}[htbp]
  \centering
  \includegraphics[width=6in,height=3in]{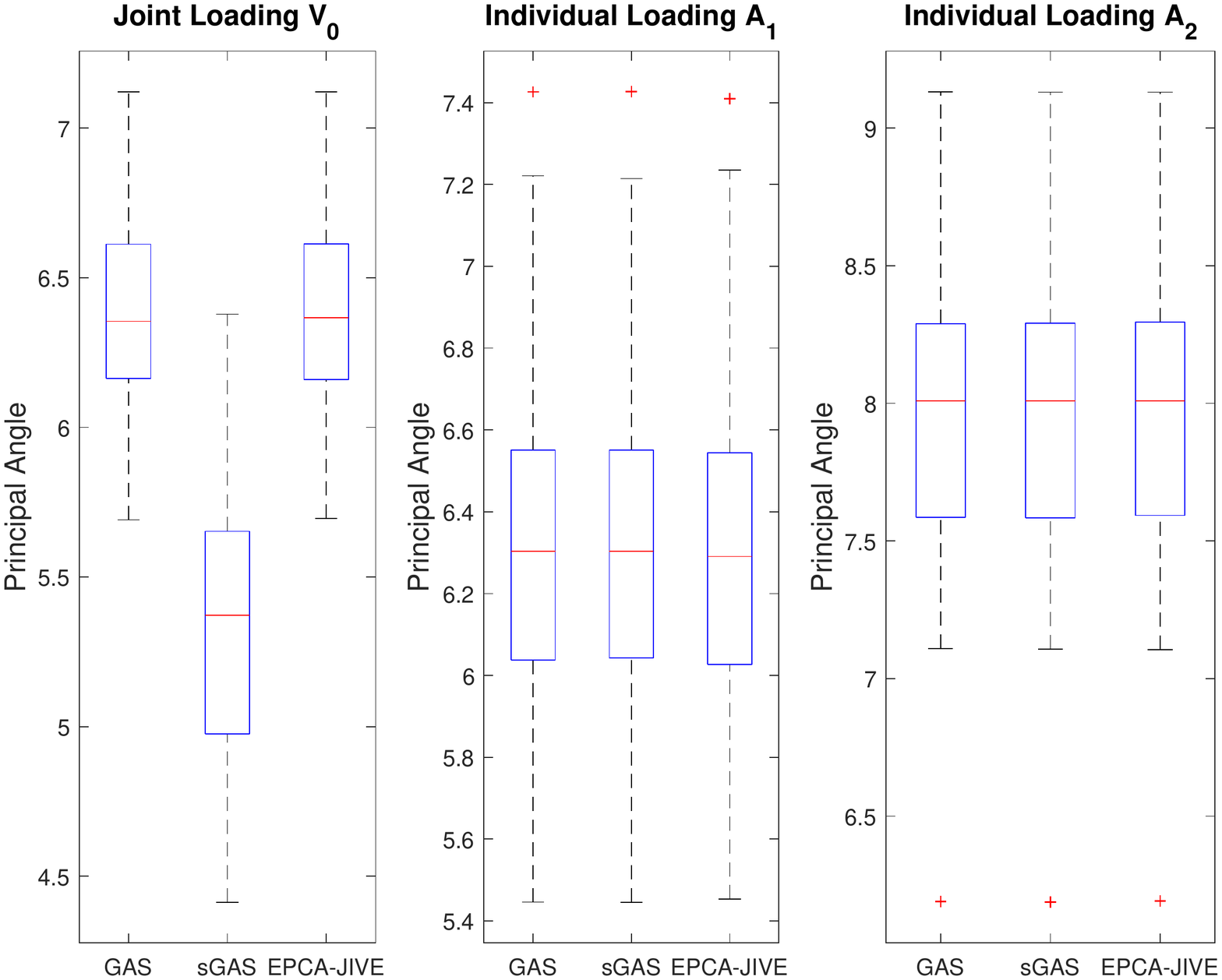}
  \caption{Sparse Setting 1 (Gaussian-Gaussian): comparison of the loading estimation accuracy among the GAS, sGAS, and EPCA-JIVE methods. From left to right, we evaluate the principal angles $\angle(\bV_0,\widehat{\bV_0}),\angle(\bA_1,\widehat{\bA_1}),\angle(\bA_2,\widehat{\bA_2})$, respectively. }\label{fig:1b}
\end{figure}

\begin{figure}[htbp]
  \centering
  \includegraphics[width=6in]{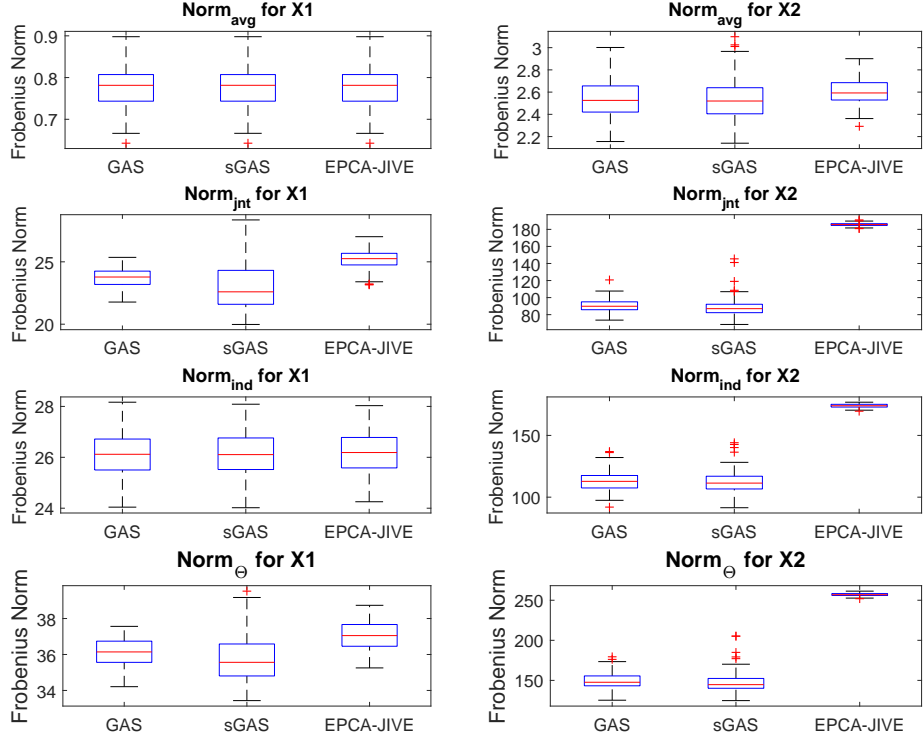}
  \caption{Sparse Setting 2 (Gaussian-Bernoulli): comparison of the low-rank structure estimation accuracy among the GAS, sGAS, and EPCA-JIVE methods. The left panels are for $\bX_1$ and the right panels are for $\bX_2$. From top to bottom, we evaluate $Norm_{avg}=\|\bmu_k-\widehat{\bmu_k}\|_\mathbb{F},\ Norm_{jnt}=\|\bU_0\bV_k^T-\widehat{\bU_0}\widehat{\bV_k}^T\|_\mathbb{F},\
Norm_{ind}=\|\bU_k\bA_k^T-\widehat{\bU_k}\widehat{\bA_k}^T\|_\mathbb{F},\ Norm_{\bTheta}=\|\bTheta_k-\widehat{\bTheta_k}\|_\mathbb{F}$, respectively. }\label{fig:2a}
\end{figure}

\begin{figure}[htbp]
  \centering
  \includegraphics[width=6in,height=3in]{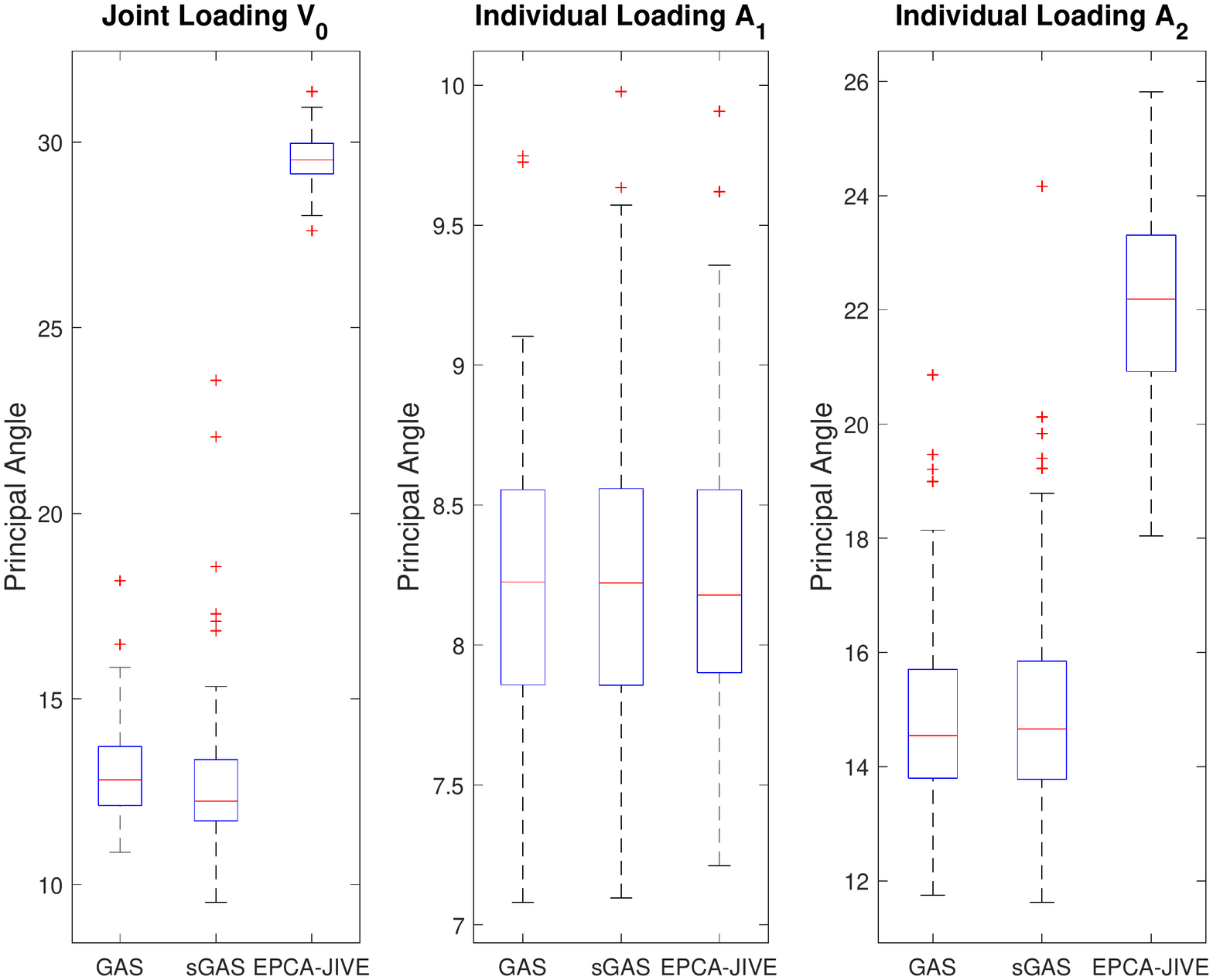}
  \caption{Sparse Setting 2 (Gaussian-Bernoulli): comparison of the loading estimation accuracy among the GAS, sGAS, and EPCA-JIVE methods. From left to right, we evaluate the principal angles $\angle(\bV_0,\widehat{\bV_0}),\angle(\bA_1,\widehat{\bA_1}),\angle(\bA_2,\widehat{\bA_2})$, respectively. }\label{fig:2b}
\end{figure}

\begin{figure}[htbp]
  \centering
  \includegraphics[width=6in]{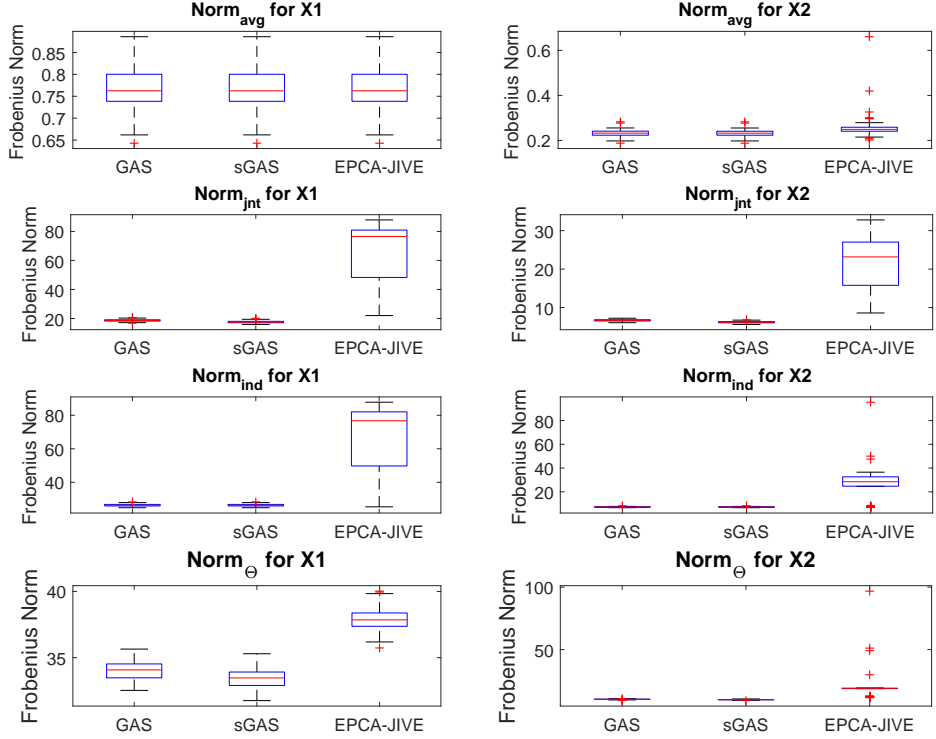}
  \caption{Sparse Setting 3 (Gaussian-Poisson): comparison of the low-rank structure estimation accuracy among the GAS, sGAS, and EPCA-JIVE methods. The left panels are for $\bX_1$ and the right panels are for $\bX_2$. From top to bottom, we evaluate $Norm_{avg}=\|\bmu_k-\widehat{\bmu_k}\|_\mathbb{F},\ Norm_{jnt}=\|\bU_0\bV_k^T-\widehat{\bU_0}\widehat{\bV_k}^T\|_\mathbb{F},\
Norm_{ind}=\|\bU_k\bA_k^T-\widehat{\bU_k}\widehat{\bA_k}^T\|_\mathbb{F},\ Norm_{\bTheta}=\|\bTheta_k-\widehat{\bTheta_k}\|_\mathbb{F}$, respectively. }\label{fig:3a}
\end{figure}

\begin{figure}[htbp]
  \centering
  \includegraphics[width=6in,height=3in]{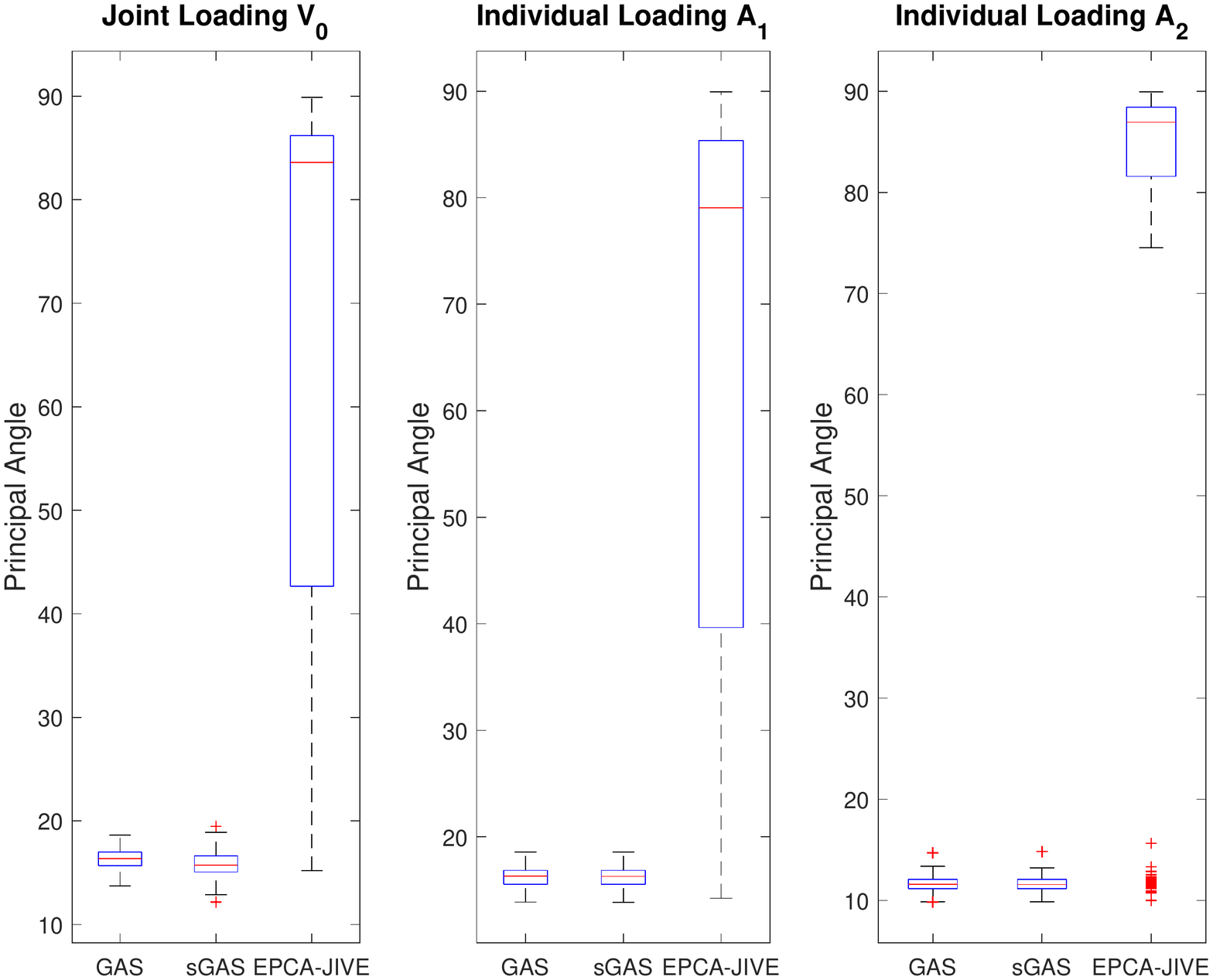}
  \caption{Sparse Setting 3 (Gaussian-Poisson): comparison of the loading estimation accuracy among the GAS, sGAS, and EPCA-JIVE methods. From left to right, we evaluate the principal angles $\angle(\bV_0,\widehat{\bV_0}),\angle(\bA_1,\widehat{\bA_1}),\angle(\bA_2,\widehat{\bA_2})$, respectively. }\label{fig:3b}
\end{figure}

\begin{figure}[htbp]
  \centering
  \includegraphics[width=6in]{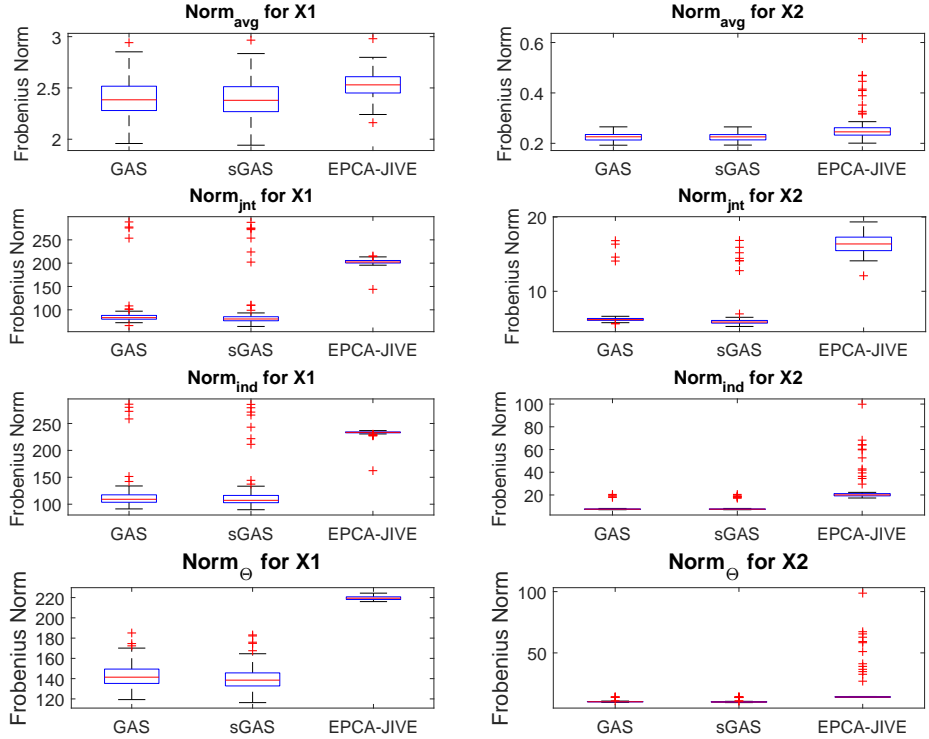}
  \caption{Sparse Setting 4 (Bernoulli-Poisson): comparison of the low-rank structure estimation accuracy among the GAS, sGAS, and EPCA-JIVE methods. The left panels are for $\bX_1$ and the right panels are for $\bX_2$. From top to bottom, we evaluate $Norm_{avg}=\|\bmu_k-\widehat{\bmu_k}\|_\mathbb{F},\ Norm_{jnt}=\|\bU_0\bV_k^T-\widehat{\bU_0}\widehat{\bV_k}^T\|_\mathbb{F},\
Norm_{ind}=\|\bU_k\bA_k^T-\widehat{\bU_k}\widehat{\bA_k}^T\|_\mathbb{F},\ Norm_{\bTheta}=\|\bTheta_k-\widehat{\bTheta_k}\|_\mathbb{F}$, respectively. }\label{fig:4a}
\end{figure}

\begin{figure}[htbp]
  \centering
  \includegraphics[width=6in,height=3in]{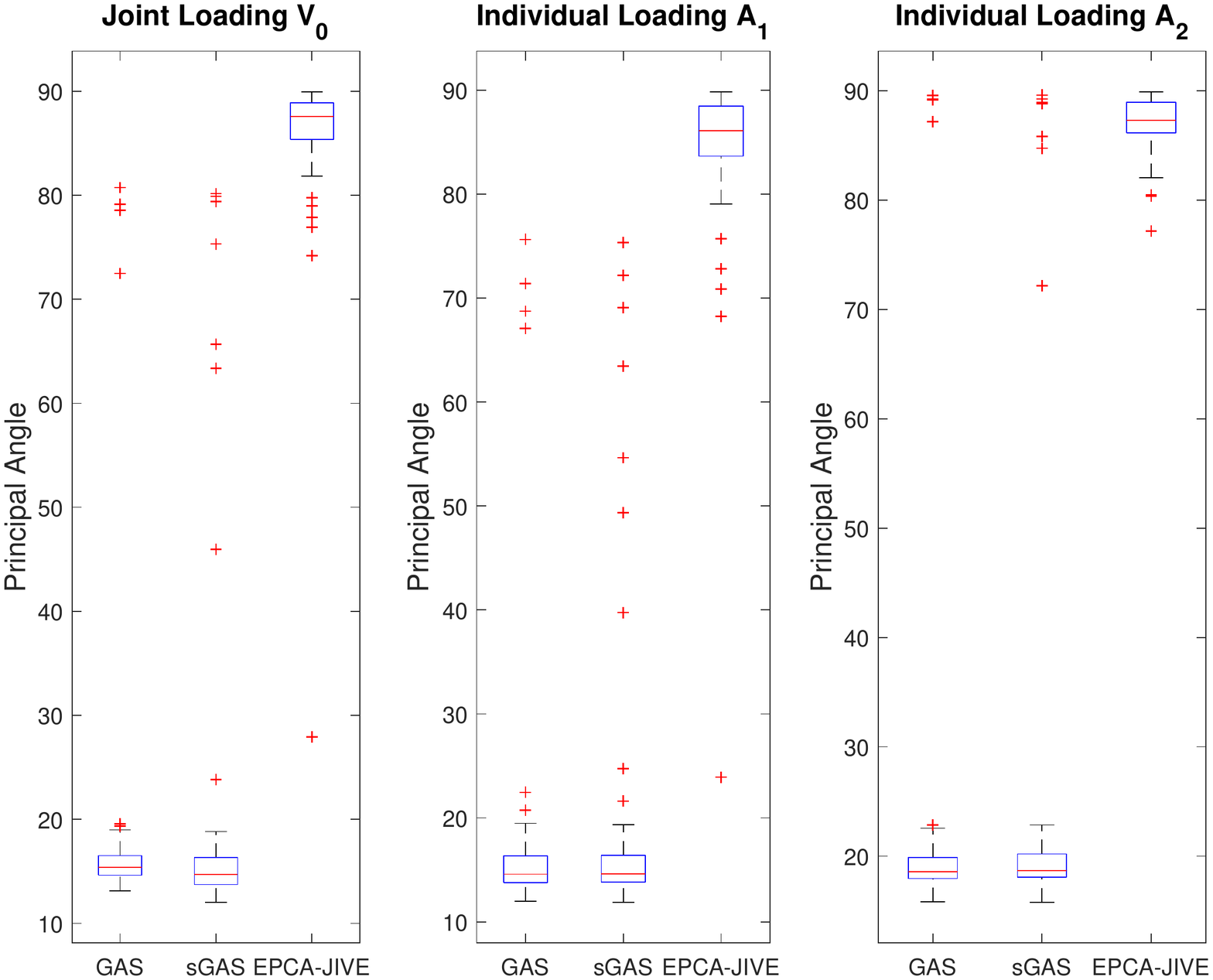}
  \caption{Sparse Setting 4 (Bernoulli-Poisson): comparison of the loading estimation accuracy among the GAS, sGAS, and EPCA-JIVE methods. From left to right, we evaluate the principal angles $\angle(\bV_0,\widehat{\bV_0}),\angle(\bA_1,\widehat{\bA_1}),\angle(\bA_2,\widehat{\bA_2})$, respectively. }\label{fig:4b}
\end{figure}

\section{Simulation under High-Dimensional Settings}
{\color{black}
In this section we investigate the effect of increasing dimensions $p_1$ and $p_2$ on the performance of the one-step GAS method. We focus on {\bf Setting 3} (Gaussian-Poisson) with $n=200$ and consider two additional variants for dimensions: $p_1=p_2=200$ and $p_1=p_2=300$.  In different settings, we keep the unit-norm scores unchanged and make the singular values proportional to the dimensions, so that the Frobenius norms of the column centered $\bTheta_k$ are proportional to the dimensions. As a result, the signal-to-noise ratios are comparable across different settings.
        We compare the relative Frobenius loss defined by $\|\bTheta_k-\widehat{\bTheta_k}\|_\mathbb{F}/\|\bTheta_k\|_\mathbb{F}$, the angles $\angle(\bV_0,\widehat{\bV_0})$ and $\angle(\bA_k,\widehat{\bA_k})$, and the computing time across different settings. The results are shown in Table \ref{tab:dim}.
The estimation accuracy assessed by the relative Frobenius loss and the principal angles becomes better with increasing $p_1$ and $p_2$  due to the ``blessing of dimensionality" \citep{li2017embracing}. While the fitting time becomes longer with higher dimensions, the model fitting procedure is still very efficient even when $p_1=p_2=300$.
}

 \begin{table}[htbp]
        \caption{Simulation results for one-step GAS under varying dimensions. Data are generated from simulation Setting 3 (Gaussian-Poisson) and its two variants with $p_1=p_2=200$ and $p_1=p_2=300$. The median and median absolute deviation (in parenthesis) of each criterion across different settings are presented. }\label{tab:dim}

        \begin{center}
        {\scriptsize\begin{tabular}{|c|cc|cc|cc|}
        \hline
        &  \multicolumn{2}{c|}{($p_1=120,p_2=120$)} & \multicolumn{2}{c|}{($p_1=200,p_2=200$)} & \multicolumn{2}{c|}{($p_1=300,p_2=300$)} \\
        \hline
        &  Data 1 & Data 2 &   Data 1 & Data 2 &  Data 1 & Data 2 \\
        \hline
        $\|\bTheta_k-\widehat{\bTheta_k}\|_\mathbb{F}/\|\bTheta_k\|_\mathbb{F}$ & 0.2894(0.0042) & 0.0254(0.0004)& 0.2071(0.0024)&0.0223(0.0002)& 0.1618(0.0018)&0.0212(0.0002)  \\
        $\angle(\bA_k,\widehat{\bA_k})$    & 15.96(0.77) & 11.49(0.55)& 12.30(0.34)&8.65(0.28)& 9.89(0.33)&7.07(0.24)\\
        \hline
        $\angle(\bV_0,\widehat{\bV_0})$    & \multicolumn{2}{c|}{16.28(0.60)}  & \multicolumn{2}{c|}{12.35(0.33)}  &\multicolumn{2}{c|}{10.35(0.30)}  \\
        Time (sec)    & \multicolumn{2}{c|}{{8.75}(0.45)}  & \multicolumn{2}{c|}{14.34(0.42)}  &\multicolumn{2}{c|}{22.92(3.47)}   \\
        \hline
        \end{tabular}}
        \end{center}
        \end{table}

\section{Simulation under Rank Misspecification}
{\color{black} We further investigate the effect of rank misspecification on the parameter estimation of the proposed method.  We focus on the simulation {\bf Setting 2} (Gaussian-Bernoulli), because its rank estimation result has some ambiguity as shown in Figure \ref{rank2}, which leaves room for rank misspecification. The true ranks are $(r_0=r_1=r_2=2)$. We particularly consider 3 additional sets of misspecified ranks: $(r_0=1,r_1=3, r_2=3)$, $(r_0=3, r_1=1,r_2=2)$, and $(r_0=4,r_1=0,r_2=2)$. The first case corresponds to the situation where a joint structure is misspecified as two individual structures (one for each data source); the second corresponds to the situation where an individual structure in the Gaussian data is misspecified as a joint structure; the third corresponds to the situation where all individual structures in the Gaussian data are misspecified as joint. We apply the GAS method with different sets of ranks to the data, and the results are shown in Table \ref{tab:rank}.

\begin{sidewaystable}[h]
\caption{Rank misspecification results for the proposed method. Data are generated from simulation Setting 2 where $r_0=r_1=r_2=2$. The median and median absolute deviation (in parenthesis) of each criterion across different rank settings are presented. For each method,  $Norm_{avg}$, $Norm_{jnt}$, $Norm_{ind}$, $Norm_{\bTheta}$ and $\angle(\bA_k,\widehat{\bA_k})$ are evaluated and compared per data set; $\angle(\bV_0,\widehat{\bV_0})$, association coefficient $\rho$, \# of iterations and computing time are evaluated across two data sets.}\label{tab:rank}
\begin{center}
{\small\begin{tabular}{|c|cc|cc|cc|cc|}
\hline
&  \multicolumn{2}{c|}{($r_0=2,r_1=2,r_2=2$)} & \multicolumn{2}{c|}{($r_0=1,r_1=3,r_2=3$)} & \multicolumn{2}{c|}{($r_0=3,r_1=1,r_2=2$)} & \multicolumn{2}{c|}{($r_0=4,r_1=0,r_2=2$)}\\
\hline
&  Data 1 & Data 2 &   Data 1 & Data 2 &  Data 1 & Data 2 &  Data 1 & Data 2\\
\hline
$\|\bmu_k-\widehat{\bmu_k}\|_\mathbb{F}$ &   {0.78}(0.04) &  {2.54}(0.10) &  0.77(0.03) & {2.57}(0.12) & (0.77(0.03))&  2.51(0.15) & 0.77(0.03) & 2.61(0.15) \\
$\|\bU_0\bV_k^T-\widehat{\bU_0}\widehat{\bV_k}^T\|_\mathbb{F}$ &   {23.69}(0.45) &  {89.36}(5.63)   & 99.90(1.18) & 213.40(3.92) & 87.13(1.44)& 98.10(6.54)  & 124.96(0.68)& 107.65(5.57) \\
$\|\bU_k\bA_k^T-\widehat{\bU_k}\widehat{\bA_k}^T\|_\mathbb{F}$ &   {26.00}(0.40) &  {110.89}(5.30) &  103.54(1.46) & 281.38(8.86) & 84.25(1.71)& 112.50(4.96)  & 120.42(0) & 115.62(5.25) \\
$\|\bTheta_k-\widehat{\bTheta_k}\|_\mathbb{F}$    & {36.08}(0.45)  & {146.86}(7.47) & 36.93(0.48) & 173.72(8.21) & 36.16(0.52) & 152.42(7.21) & 36.11(0.48) & 162.16(6.80) \\
$\angle(\bA_k,\widehat{\bA_k})$    &{8.18}(0.40)  &  {14.47}(0.69)  & 8.03(0.38)  & 14.64(0.96) & 7.28(0.33) & 14.39(0.90)  & NA & 14.78(1.05) \\
\hline
$\angle(\bV_0,\widehat{\bV_0})$    & \multicolumn{2}{c|}{12.96(0.79)}  & \multicolumn{2}{c|}{12.21(0.84)}  &\multicolumn{2}{c|}{12.91(0.79)}  &\multicolumn{2}{c|}{12.84(0.75)} \\
$\rho$ & \multicolumn{2}{c|}{0.5612(0.0046)}  & \multicolumn{2}{c|}{0.5544(0.0055)}  &\multicolumn{2}{c|}{0.5889(0.0059)}  &\multicolumn{2}{c|}{0.6178(0.0044)}  \\
\# iteration & \multicolumn{2}{c|}{21(1.00)}  & \multicolumn{2}{c|}{25(2.00)}  &\multicolumn{2}{c|}{20(1.00)}  & \multicolumn{2}{c|}{14(1.00)}\\
Time (sec)    & \multicolumn{2}{c|}{{10.94}(1.36)}  & \multicolumn{2}{c|}{14.07(1.90)}  &\multicolumn{2}{c|}{11.90(0.85)}  & \multicolumn{2}{c|}{3.96(0.23)}  \\
\hline
\end{tabular}}
\end{center}
\end{sidewaystable}

We observe that the Frobenius losses of individual structures and joint structures estimated under misspecified ranks are larger than those estimated under the true ranks. This is expected because some individual structures might be mistaken as joint structures and vice versa. Nevertheless, the Frobenius losses of the estimated natural parameter matrices and the principal angles for respective loadings are comparable across different rank settings. Moreover, the association coefficients estimated under different ranks are relatively stable. The results demonstrate that the GAS method and the corresponding association coefficient are both robust against rank misspecification.  }

\end{document}